\documentclass[a4paper, 12pt]{article}
\usepackage{amsmath,amsthm,amssymb,graphicx}
\usepackage[margin=2.5cm]{geometry}
\usepackage{bbm}
\usepackage{xparse}
\usepackage{color}
\usepackage{enumitem}
\usepackage[square]{natbib}
\setcitestyle{numbers}

\usepackage{epstopdf}
\usepackage{chngcntr}
\counterwithin{figure}{section}

\theoremstyle{plain}
\theoremstyle{definition}
\numberwithin{equation}{section}

\newtheorem{definition}{Definition}[section]
\newtheorem{theorem}[definition]{Theorem}
\newtheorem{assumption}[definition]{Assumption}
\newtheorem{remark}[definition]{Remark}

\newtheorem*{claim*}{Claim}
\newtheorem{lemma}[definition]{Lemma}
\newtheorem{corollary}[definition]{Corollary}
\newtheorem{proposition}[definition]{Proposition}

\newtheorem{condition}[definition]{Condition}

\newcommand{\E}{\mathbb{E}}
\newcommand{\D}{\mathbb{D}}
\newcommand{\F}{\mathbb{F}}

\newcommand{\R}{\mathbb{R}}

\renewcommand{\S}{\mathbb{S}}
\newcommand{\bs}{\mathbf{s}}
\newcommand{\Q}{\mathbb{Q}}

\newcommand{\N}{\mathbb{N}}
\renewcommand{\P}{\mathbb{P}}

\newcommand{\XX}{\mathcal{X}}

\newcommand{\PP}{\mathcal{P}}

\newcommand{\FF}{\mathcal{F}}
\newcommand{\MM}{\mathcal{M}}

\newcommand{\BB}{\mathcal{B}}
\renewcommand{\AA}{\mathcal{A}}

\newcommand{\ZZ}{\mathcal{Z}}

\newcommand{\set}{\mathfrak{P}}

\newcommand{\po}{c}

\newcommand{\vmu}{\vec{\mu}}
\newcommand{\MMN}{\MM_{\vmu}}
\newcommand{\MMNS}{\MM_{\vmu}^-}
\newcommand{\MMU}{\MM_{\vmu,\set}}
\newcommand{\MMUS}{\MM_{\vmu,\set}^-}
\newcommand{\lMMUS}{\MM_{\vmu,\set}^{\textrm{loc}}}
\NewDocumentCommand{\MMUE}{O{\vmu}O{\set}m}{\MM_{{#1},{#2},{#3}}}

\NewDocumentCommand{\Dual}{O{\vec{\mu}}O{\set}m}{
V_{#1,#2}^{#3}
 }

\newcommand{\duals}{V_{\vec{\mu},\set}}
\newcommand{\primalp}{P_{\XX,\PP,\set}}
\newcommand{\primalpc}{P_{\XX_c,\PP,\set}}
\newcommand{\primalpp}{P_{\XX_p,\PP,\set}}
\newcommand{\primalppd}{P^d_{\XX_p,\PP,\set_{\vec{T}}}}
\newcommand{\dualp}{V_{\XX,\PP,\set}}
\newcommand{\dualpc}{V_{\XX_c,\PP,\set}}
\newcommand{\dualpp}{V_{\XX_p,\PP,\set}}
\newcommand{\dualppd}{V^d_{\XX_p,\PP,\set_{\vec{T}}}}
\newcommand{\primals}{P_{\vec{\mu},\set}}

\NewDocumentCommand{\Lip}{O{}}{\mathfrak{G}_{#1}(\R)}
\NewDocumentCommand{\Lips}{O{N}O{M}}{\mathfrak{G}_{#1}([-#2,#2])}

\NewDocumentCommand\OneOptWithDefault{O{default}m}{Text using #1 (could be the default) and #2}
\newcommand{\td}{\mathrm{d}}

\newcommand{\Lin}{\mathrm{Lin}}

\newcommand{\indicator}[1]{\mathbbm{1}_{\left\{ {#1} \right\} }}
\newcommand{\indicators}[1]{\mathbbm{1}_{{#1}}}

\newcommand\restr[2]{{
  \left.\kern-\nulldelimiterspace 
  #1 
  \vphantom{\big|} 
  \right|_{#2} 
  }}

\usepackage[
bookmarks,
bookmarksopen=true,
pdftitle={},
pdfauthor={AMG Cox, Zhaoxu Hou, Jan Obloj},
pdfcreator={AMG Cox, Zhaoxu Hou, Jan Obloj},
pdfsubject={},
pdfkeywords={},
colorlinks=false,
anchorcolor=black,
filecolor=magenta, 
menucolor=red, 
linkcolor=black, 
]{hyperref}

\title{Robust pricing and hedging under trading restrictions and the
  emergence of local martingale models\thanks{Zhaoxu Hou gratefully acknowledges the support of the Oxford-Man Institute of Quantitative Finance and Balliol College in Oxford. Jan Ob\l\'oj gratefully acknowledges funding received from the European Research Council under the European Union's Seventh Framework Programme (FP7/2007-2013) / ERC grant agreement no. 335421. Jan Ob\l\'oj is also thankful to the Oxford-Man Institute of Quantitative Finance and St John's College in Oxford for their financial support. }}
\author{
Alexander M.~G.~Cox\thanks{University of Bath, email:
  \texttt{a.m.g.cox@bath.ac.uk}, web: \url{http://www.maths.bath.ac.uk/}$\sim$\url{mapamgc}}
\and Zhaoxu Hou\thanks{University of Oxford, email: \texttt{zhaoxu.hou@maths.ox.ac.uk}, web: \url{http://www.maths.ox.ac.uk/people/profiles/zhaoxu.hou}}
\and Jan Ob\l{}\'{o}j\thanks{University of Oxford, email: \texttt{Jan.Obloj@maths.ox.ac.uk}, web: \url{http://www.maths.ox.ac.uk/people/profiles/jan.obloj}.}
        }
\date{\today}
\RequirePackage[normalem]{ulem} 
\RequirePackage{color}\definecolor{RED}{rgb}{1,0,0}\definecolor{BLUE}{rgb}{0,0,1} 

\begin{document}

\maketitle

\begin{abstract}
We pursue the robust approach to pricing and hedging in which no probability measure is fixed but call or put options with different maturities and strikes can be traded initially at their market prices. We allow for inclusion of robust modelling assumptions by specifying the set of feasible paths. In a discrete time setup with no short selling, we characterise absence of arbitrage and show that if call options are traded then the usual pricing-hedging duality is preserved. In contrast, if only put options are traded a duality gap may appear. Embedding the results into a continuous time framework, we show that the duality gap may be interpreted as a financial bubble and link it to strict local martingales. This provides an intrinsic justification of strict local martingales as models for financial bubbles arising from a combination of trading restrictions and current market prices. 

\end{abstract}

\section{Introduction}

The approach to pricing and hedging of options through considering the dual problem of finding the expected value of the payoff under a risk-neutral measure is both classical and well understood. In a complete market setting it is simply the way to compute the hedging price, as argued by \citet{BlackScholes:73}. In incomplete markets, the method originated in \citet{EKQuenez}, culminating in the seminal work of \citet{FTAP}. Almost as classical is the problem of finding superhedging prices under various constraints on the set of admissible portfolios. Questions of this type arise in \citet{CvitanicKaratzas}, where convex constraints in the hedging problem lead to a dual problem where one looks for the largest expectation of the payoff of the derivative in a class of {\it auxiliary} markets, where the auxiliary markets are a modification of the original markets reflecting the trading constraints. In the special case of markets where participants may not short sell assets, the class of auxiliary markets correspond to the class of supermartingale measures. (For further results in this direction, see e.g.\ \citet{JouiniKallal,CPT,PT,Pulido}.)

In this paper, we combine trading constraints with concepts from robust derivative pricing. In robust pricing, one aims to minimise modelling assumptions by not pre-supposing the existence of a given probabilistic model for the underlying assets. Instead, we replace modelling assumptions through two weaker assumptions: first, we suppose that our observed realisation of the price process will lie in some set $\set$ of possible outcomes, e.g.~the set of paths whose sum of squared differences is bounded by some given constant, or the continuous time analogue, the set of paths with quadratic variation bounded by a given constant; second, we suppose that there are additional options which may be traded at time zero, for prices which are observed in the market. In this paper we will assume that the additional options which are traded are either European put or call options. In particular, we will suppose that at fixed maturity dates, the prices of all puts/calls on the underlying are known. The presence of put and call options fixes the law of the process under the risk-neutral measure in any calibrated model --- this is a fact first observed by \citet{breeden1978prices} --- and hence constrains the set of probability measures over which we optimise. There has recently been substantial interest in robust pricing problems with a literature that can be traced back to the seminal paper of \citet{Hobson}. The results in this paper are based on the discrete-time approach of \citet{mass_transport}, where a duality result is shown using concepts from optimal transport. 

In the discrete time setting, our results can be summarised as follows: we suppose that we are given a sequence of call price functions at maturity dates $T_1 < T_2 < \dots < T_n$. We show that these prices are consistent with absence of certain natural types of arbitrage if and only if they give rise to a sequence of probability measures $\mu_1, \dots, \mu_n$ on $\R^n_+$, which satisfy natural ordering properties. These then, as explained above, correspond to the implied marginal distributions of the asset under feasible risk-neutral measures. (Note that here and throughout, we assume that all assets are denominated in units of some numeraire, for example discounted by the money market account). Classically, the measures would be in convex order. However in the absence of the ability to short-sell the asset, it is not possible to generate an arbitrage when $m_k = \int x \mu_k(dx) > \int x \mu_{k+1}(dx) = m_{k+1}$, and so the expected value of the asset according to the (implied) risk-neutral measure may be smaller at later maturities. We then show that the minimal price of a portfolio involving call options and long positions in the asset, and which superhedges a derivative for every path in $\set$, is equal to the supremum of the expected value of the derivative's payoff, where the supremum is taken over all supermartingale measures which have full support on $\set$, and under which the law of the asset at $T_k$ is equal to $\mu_k$. This result generalises Corollary~1.1 in \citet{mass_transport} by including a restriction to a certain set of paths $\set$ and a short-selling constraint. Observe also that, in the case where the measures $\mu_k$ all have the same mean, which is equal to the initial stock price $s_0$, the class of supermartingale measures is simply the class of martingale measures.

We also consider the case where the set of call options is replaced by put options with the same maturities. Since short selling of the asset is not permitted, one cannot immediately compare to the case where the call options are available to trade, even if the set of possible implied marginal laws remains the same. In this case we show that a duality gap arises when the initial asset price $s_0$ is strictly larger than the implied mean $m_k$ for some maturity $T_k$. In particular, there is no longer equality between the cheapest superhedge and the largest model-consistent price --- rather, we see a difference which can be characterised in terms of the limit behaviour of the put prices as the strike goes to infinity.

The easiest example of this duality gap arises in considering the difference between the implied price of a forward contract written on the asset --- taking the forward to be a contract which pays the holder the value of the asset at some future date $T_k$, then the forward contract will have a model-implied price $m_k = \int x \mu_k(dx)$, which, in the cases of interest, will be strictly smaller than the initial price of the asset $s_0$. In the case where call options are traded, the forward may be superhedged for $m_k$ using call options (the call option with strike $0$ has the same payoff as the forward). In the case where put options are traded, this is not the case --- instead, the cheapest super-replicating strategy will simply be to purchase the asset at time 0, which has cost $s_0$.

Historically, there has been relatively little study of asset prices which are strict supermartingales\footnote{A strict supermartingale is a supermartingale which is not a martingale. Since we only consider non-negative processes over finite time-horizons, a strict supermartingale is therefore a supermartingale which has a non-constant expected value. Similarly, a strict non-negative local martingale is a local martingale which is also a strict supermartingale.} under the risk-neutral measure in the literature. Their main appearance has been as models for the study of financial bubbles, where strict local-martingales are considered. We believe that our results, both in discrete time and in continuous time, contribute to and provide a novel perspective on the existing literature on financial bubbles.

In mathematical finance, the modelling of financial bubbles using local-martingale models can be traced back to \citet{HLW}, with subsequent contributions including \citet{CH,JPS,JPS2}. Before \citet{HLW}, a number of authors observed that, in certain circumstances, models which were only strict local martingales arise naturally and/or are interesting of their own right (and can be attributed some financial interpretation); see \citet{Lewis,DS,LW,Sin}. One of the most common examples of a naturally occurring class of local-martingale models is the class of CEV models, $dS_t = S_t^\alpha dB_t, S_0 = s_0$, where $\alpha>1$. In the case where $\alpha = 2$, one recovers the inverse of a 3-dimensional Bessel process, which was studied in \citet{DS}. More recently, the class of Quadratic Normal Volatility (QNV) models have also been studied, which are mostly strict local-martingales, but typically calibrate well to market data; see \citet{CFR}.

We build our contribution to this literature by embedding the discrete time results into a continuous time framework. Consider a continuous time framework with dynamic trading in the asset and call or put options traded initially for certain fixed maturities. Then the discrete setup is naturally included by considering trading strategies which only rebalance at the maturity dates of the options. Discrete time supermartingale measures are obtained as projections of local-martingale measures which meet the given marginals. The duality gap is preserved when put options trade, and this gap has a possible interpretation as a financial bubble. To make this generalisation, it is necessary to introduce a pathwise superhedging requirement which enforces a collateral requirement. A similar requirement has already been considered in \citet{CH}.
We therefore believe that an important consequence of this paper is the following interpretation of local-martingale models in financial applications: {\it local-martingale models naturally arise due to trading constraints.} 

This has an impact on the existing literature on financial bubbles: intrinsically, we believe that models where asset prices are strict local martingales (under a risk-neutral measure) are models which arise due to constraints on possible trading strategies. They thus correspond exactly to \emph{rational} or \emph{speculative bubbles} in the asset pricing and economics literature. These are usually driven by short-selling constraints and/or disagreement between the agents on the fundamental values due to heterogenous beliefs or overconfidence, see \citet{hugonnier}, \citet{HarrisonKreps:78} and \citet{ScheinkmanXiong:03}.  Strict local martingale models are a very natural class of models for bubbles, since there is a natural notion of a `fundamental' price which diverges from the traded price. However, as we show, this divergence is `rational' and driven by the absence of arbitrage combined with trading restrictions, as in speculative bubbles.  This is different from the case of an `irrational bubble' when divergence between the market price of an asset and its fundamental price is driven by some behavioural aspect of market participants, rather than specific market features. In this sense, an important contribution of this article for the literature on bubbles is to divorce any notions of `irrationality' from the financial study of strict local martingale models.

We also make the observation that, although we present results on local-martingale models in continuous time, our approach is firmly rooted in a discrete-time setup, and all pricing results in continuous time follow essentially from the  corresponding discrete-time results. One interpretation is that these models therefore are the natural discrete-time analogues of local-martingale models (in this sense, our results provide an alternative response to the first criticism discussed in \citet{Protter}; see also \citet{JP}). However, it seems to us that the implication more naturally runs in the other direction: in discrete-time, our models are very natural, and easily specified. In continuous-time, however, local-martingales are very subtle processes, and the difference between a local-martingale and a martingale is not easy to detect --- our paper provides a clear specification of a discrete-time setup which could be interpreted in continuous time as a local-martingale model. As a result, in our setup local-martingale phenomena arise naturally, and reflect specific market conditions. This contrasts with the arguments of e.g.~\citet{GR}, who argue against local-martingale models on the basis that they can always be approximated by martingale models.

Short selling bans as a regulatory tool to discourage speculation and stabilise markets have proved to be popular among emerging markets and during times of financial crisis. During the U.S. subprime mortgage crisis, short selling of 797 financial stocks in U.S. markets was banned by the SEC between September 19, 2008 and October 8, 2008. Around the same time, the South Korean Financial Supervisory Commission imposed an outright prohibition of short selling of any listed stocks in an attempt to curb the spread of malignant rumours in the market. The ban was lifted for non-financial stocks about a year later, while the constraints on financial stocks remained until November 2013. Interestingly, the U.S. and South Korea both have very active derivatives markets and, in both examples, the bans on short selling did not extend to derivative markets. This allowed market makers and investors to use options to hedge portfolios and express pessimistic views. In light of a series of short selling
bans across the globe, the question of their impact on stocks and derivatives markets is once again a matter of concern to academics and policy makers, see e.g.\ \citet{battalio2011regulatory}, \citet{hendershott2013intended}; the current paper represents a theoretical contribution to this literature. \citet{battalio2011regulatory} study the U.S. short selling ban in September 2008 and find that synthetic share prices for banned stocks, computed separately for puts and for calls, become significantly lower than the actual share prices, accompanied by increases in bid-ask spreads. The findings correspond to the setting of our paper with $m_k<s_0$, making it particularly interesting.

Finally, we note that in parallel to our research \citet{FahimHuang:14} and \citet{BayraktarZhou:14} considered discrete time robust pricing and hedging with trading restrictions. \citet{FahimHuang:14} use concepts from optimal martingale transport but assume market input in form of distributions $\mu_i$ already satisfying a set of assumptions which in our paper are characterised in terms of arbitrage opportunities. \citet{BayraktarZhou:14} adopt the quasi-sure analysis of \citet{BouchardNutz:13} with finitely many traded options. As a result, in both cases the pricing-hedging duality holds and no links are made to modelling of financial bubbles in discrete or continuous time. The focus of both papers is on general convex portfolio constraints.

This paper is organised as follows. Section \ref{sec:robustmodelling} discusses the robust modelling framework in discrete time. Sections \ref{sec:dt_calls} and \ref{sec:dt_puts} specialise respectively to the case when call or put options are traded. The latter in particular explores when a duality gap arises. Subsequently Section \ref{sec:continuous} focuses on the continuous time setup. Several proofs are relegated to the Appendix.



\section{Robust framework for pricing and hedging}
\label{sec:robustmodelling}
We consider a financial market with two assets: a risky asset $S$ and a numeraire (e.g.\ the money market account). 
All prices are denominated in the units of the numeraire. In particular, the numeraire's price is thus normalised and equal to one. We assume initially that $S$ is traded discretely in time at maturities $0=T_0<T_1<T_2<\ldots<T_n=T$. This is extended to a continuous time setup in Section \ref{sec:continuous}. 
The asset starts at $S_0=s_0>0$ and is assumed to be non-negative. 
We work on the canonical space with a fixed starting point $\Omega = \{(\omega_0,\ldots,\omega_n)\in \R_+^{n+1} : \omega_0=s_0\}$. The coordinate process on $\Omega$ is denoted $\S=(\S_i)_{i=0}^n$ i.e. 
\begin{align*}
\S_i\,:\, \Omega\to \R_+,\, \S_i(\omega_0,\omega_1,\ldots,\omega_n)=\omega_i,\, i=0,\ldots,n,
\end{align*}
and $\F=(\FF_{i})_{i=1}^n$ is its natural filtration, $\FF_{i}=\sigma(\S_0,\ldots,\S_i)$ for any $i=0,\ldots,n$. 

We pursue here a robust approach and do not postulate any probability measure which would specify the dynamics for $S$. Instead we assume that there is a set $\XX$ of market traded options with prices known at time zero, $\PP(X)$, $X\in \XX$. The trading is frictionless and options in $\XX$ may be bought or sold at time zero at their known prices. Hence we extend $\PP$ to be a linear operator defined on 
$$\Lin(\XX)= \left\{a_0+\sum_{i=1}^m a_i X_i: m\in \mathbb{N}, a_0,a_i\in \R, X_i\in \XX \textrm{ for all } i=1,\ldots,m\right\}.$$
As explained above, the numeraire has a constant price equal to one. Finally, the risky asset $S$ may be traded at any $T_i$, $i=0,\ldots, n$, \emph{however no short selling is allowed}.

We will consider two cases: when $\XX$ is composed of call options or of put options:
\begin{align*}
\XX_c=\{(\S_{i}-K)^+: i=1,\ldots,n,\,\, K\in \R_+\},\quad \XX_p=\{(K-\S_{i})^+: i=1,\ldots,n,\,\, K\in \R_+\}.
\end{align*}
An admissible (semi-static) trading strategy is a pair $(X,\Delta)$ where $X\in \Lin(\XX)$ and $\Delta=(\Delta_j)$ are bounded non-negative measurable functions $\Delta_j:\R^j_{+}\to \R_+$, $j=0,\ldots,n-1$. The total payoff associated to $(X,\Delta)$ is given by 
\begin{align*}
\Psi_{X,\Delta}(\S):=X(\S)+\sum_{j=0}^{n-1}\Delta_j(\S_1,\ldots,\S_j)(\S_{j+1}-\S_j).
\end{align*}
The cost of following such a trading strategy is equal to the cost of setting up its static part, i.e.\ of buying the options at time zero, and is equal to $\PP(X)$.
We denote the class of admissible (semi-static) trading strategies by $\AA_\XX$. We write $\AA_c$ (resp.\ $\AA_p$) for the case $\XX=\XX_c$ (resp.\ $\XX=\XX_p$).
Note that since no short selling is allowed these are genuinely different and, as we will see, will give very different results. Indeed, note that in the former the short selling of call options is allowed, including the strike zero i.e.\ the forward, providing a super-replication of the asset $S$, possibly at a strictly cheaper price than $s_0$. This feature is not present when dealing with put options. 

We are interested in characterising and computing superhedging prices. All the quantities we have introduced are defined pathwise and the superhedging property is also required to hold pathwise. We have also only made mild assumptions on the market mechanisms (e.g.\ no frictions) but no specific modelling assumptions on the dynamics of the assets. A natural way to incorporate beliefs into the robust framework is through specifying the set $\set\subset \Omega$ of ``possible paths", i.e.\ paths we deem feasible and for which the hedging strategies are required to work. This can be thought of as specifying the maximal support of the plausible models. In this way, with the support ranging from all paths to e.g.\ paths in a binomial model, the robust framework can interpolate between model-independent and model-specific setups. The set $\set$ might be obtained through time series analysis of the past data combined with modelling and a given agent's idiosyncratic views and is referred to as the \emph{prediction set}. Note that since there is no probability measure specified and hence no distinction between the \emph{real} and the \emph{risk-neutral} measure, it is very natural to combine two streams of information: time-series of past data and forward-looking option prices. This idea goes back to \citet{Mykland_prediction_set} and we refer to \citet{NadtochiyObloj:14} 
for more details and extended discussion. 

We call the triplet $(\XX,\PP,\set)$ of prediction set, market traded options $\XX$ and their prices, the \emph{robust modelling inputs}. The fundamental financial notions defined below, e.g.\ the arbitrage or the super-replication price, are implicitly relative to these inputs. 

\begin{definition} \label{def:superreplication}
The \emph{super-replication cost} of a derivative with payoff $G:\Omega\to \R$, denoted by $\dualp(G)$, is the smallest initial capital required to finance a semi-static trading strategy which super-replicates $G$ for every path in $\set$, i.e.
\begin{align}\label{eq:dual_robust_hedging_cost}
\dualp(G):=\inf\Big\{\PP(X)\,:\, \exists (X,\Delta)\in \AA_\XX \text{ s.t } \Psi_{X,\Delta}\ge G\,\text{ on }\set\Big\}.
\end{align}
\end{definition}
Note that since $\omega_0=s_0$ for all $\omega\in \Omega$, it is equivalent to see $G$ as a function from $\Omega$ or from $\R_+^n$. We will be tacitly switching between these viewpoints, the former is used when writing $G(\S)$, the latter when imposing conditions on $G$, see e.g.\ \eqref{eq:Gassumption} below.

Our aim is to understand when a pricing-hedging duality holds, i.e.\ when the super-replication price can be computed through the supremum of expectations of the payoff over a suitable class of probabilistic models. 
\begin{definition}\label{def:marketmodel}
A market calibrated model is a probability measure $\P$ on $(\Omega,\F)$ satisfying $\P(\set)=1$ and for any $(X,\Delta)\in \AA_\XX$ 
\begin{align}
\E_{\P}[\Psi_{X,\Delta}(\S)]\le\PP(X),\label{eq:def_market_model_1}
\end{align}
where here, and throughout, we make the convention that $\infty-\infty=-\infty$ so that the LHS in \eqref{eq:def_market_model_1} is always well defined. The set of market calibrated models is denoted $\MM^-_{\XX,\PP,\set}$.
\end{definition}
\begin{remark}\label{rk:primallessthandual}
It follows from the definition that if $\MM^-_{\XX,\PP,\set}\neq \emptyset$ then for any Borel function $G:\Omega\to \R$
\begin{equation}\label{eq:primallessthandual}
\primalp(G):=\sup_{\P\in \MM^-_{\XX,\PP,\set}}\E_{\P}[G(\S)] \leq \dualp(G).
\end{equation}
We sometimes refer to the LHS of the above inequality as the \emph{primal value} and to the RHS as the \emph{dual value}. This duality gap provides us with two different notions for the `price' of the asset. Historically, this has be used as a method for modelling certain financial phenomena where more that one price may appear. A natural interpretation which arises in the literature is to say that the superhedging price $\dualp(G)$ represents the \emph{market price} of $G$, and the primal side, which represents the worst model price, among models consistent with the prices observed in the market, can be thought of as the \emph{fundamental price} of $G$. The case of strict inequality in \eqref{eq:primallessthandual} then admits an interpretation as a \emph{financial bubble} --- that is, a difference between the market price, and the fundamental price. We will consider this question further below.
\end{remark}
\begin{remark}\label{rk:marketmodels}
It follows from the definition that under any market calibrated model $\P$, the canonical process $\S=(\S_i)_{i=1}^{n}$ is a supermartingale. Such a measure is called a supermartingale measure. Furthermore, for any $X\in \XX$, \eqref{eq:def_market_model_1} holds for both $(X,0)$ and $(-X,0)$ so $\P$ is calibrated to options in $\XX$, i.e. $\P$ satisfies $\E_{\P}[X]=\PP(X)$ for any $X\in \XX$.
\end{remark}

\begin{definition}\label{def:no_arbitrage}
We say there is a \emph{robust uniformly strong arbitrage} if there exists a trading strategy $(X,\Delta)\in \AA_\XX$ with a negative price $\PP(X)< 0$ and a non-negative payoff $\Psi_{X,\Delta}(\bs)\ge 0$ for all $\bs\in \set$.
\end{definition}

\begin{remark}\label{rk:strongarbitrage}
By definition, the market admits a robust uniformly strong arbitrage if and only if $\dualp(0)<0$. 
When $\set=\Omega$ the notion of \emph{robust uniformly strong arbitrage} corresponds to the \emph{model independent arbitrage}, see \citet{davis_hobson2007} and \citet{cox2011robust}. In a general robust setting, the existence of an arbitrage in the above sense may depend on the modelling assumptions, expressed through $\set\subset \Omega$, which justifies the terminology. Equally, we stress that the arbitrage is required to be uniform in outcomes $\omega\in \set$ to distinguish from a slightly weaker notion, used in \citet{acciaio2013model}, of a strategy $(X,\Delta)\in \AA_\XX$ with $\PP(X)\leq 0$ and a positive payoff $\Psi_{X,\Delta}(\bs)> 0$. At present it is not clear to us if these two notions are equivalent. We can show that they are equivalent in our setup when $\XX=\XX_p$ or when $\XX=\XX_c$ and either $\set=\Omega$ or property (iii) in Condition \ref{cond:market set-up} below holds\footnote{We chose not to present these argument here as they are lengthy and deviate from the main focus of our paper. Instead they will be part of a separate note on arbitrage in a robust setting, available on request or through the authors' webpages.}.
\end{remark}

The three papers mentioned in the above remark also showed that typically absence of robust uniformly strong arbitrage is not sufficient to guarantee a (robust) fundamental theorem of asset pricing holds and introduced various weaker notions or additional assumptions. Here we follow \citet{cox2011robust}:
\begin{definition}
We say that there is a weak free lunch with vanishing risk (WFLVR) if there exist trading strategies $(X_k,\Delta_k)\in \AA_\XX$ and $(X,\Delta)\in \AA_\XX$ such that $\Psi_{X_k,\Delta_k}\to 0$ pointwise on $\set$, $\lim_k\PP(X_k)$ is well defined with $\lim_k\PP(X_k)< 0$ and $\Psi_{X_k,\Delta_k}\ge \Psi_{X,\Delta}$.
\end{definition}
Note that the requirement that $\lim_k \PP(X_k)$ exists is made with no loss of generality as we could always pass to a subsequence of strategies. 
Note also that a \emph{robust uniformly strong arbitrage} is by definition also a \emph{WFLVR}. A version of the Fundamental Theorem of Asset Pricing in our context, given below in Proposition \ref{proposition:no_arbitrage} for the case $\XX=\XX_c$ and in Proposition \ref{proposition:no_arbitrage put} for the case $\XX=\XX_p$, states that absence of WFLVR is equivalent to existence of a calibrated market model. Further, as in \citet{davis_hobson2007} and \citet{cox2011robust}, we can characterise absence of WFLVR through the properties of $\PP$. It is this feature which prevents us from considering, as in e.g.\  \citet{acciaio2013model}, shifted payoffs $X-\PP(X)$ for traded options and eliminating $\PP$ from the discussion. 

\section{Robust pricing--hedging duality when call options trade}\label{sec:dt_calls}
In this section we consider the market in which call options are traded, i.e.\ $\XX=\XX_c$. Our main result states that we recover the duality known from the case when short selling restrictions are not present. Throughout we assume that $\set$ is a closed subset of $\Omega$.

\subsection{Market input and no arbitrage}
We start by establishing a robust Fundamental Theorem of Asset Pricing for our setting which links absence of arbitrage, properties of call prices and existence of a calibrated market model.
\begin{condition}\label{cond:market set-up} 
Let $\XX=\XX_c$ and $c_i(K):=\PP((\S_i-K)^+)$, $i=1,\ldots, n$, $K\geq 0$. Then
\begin{enumerate}
\item[(i)] \label{item:market set-up 1} $c_i(x)$ is a non-negative, convex, decreasing function of $x$ on $\R_+$, 
\item[(ii)] \label{item:market set-up 2}$s_0\ge c_1(0)\ge\ldots\ge c_n(0)\ge 0$ and $c^{\prime}_{i}(0+)\ge -1$, 
\item[(iii)] \label{item:market set-up 3}$c_i(K)\to 0$ as $K\to \infty$, 
\item[(iv)] For any $x\in\R_+$, $c_i(0)-c_i(x)$ is non-increasing in $i$.\label{item:market set-up 4}
\end{enumerate}
\end{condition}

A robust Fundamental Theorem of Asset Pricing in our setup reads as follows.
\begin{proposition}\label{proposition:no_arbitrage}
Suppose $\set$ is a closed subset of $\Omega$ and $\XX=\XX_c$. Then there is no WFLVR if and
only if there exists a market calibrated model, which then implies 
Condition~\ref{cond:market set-up} holds. Furthermore, if $\set=\Omega$ then Condition~\ref{cond:market set-up} implies absence of WFLVR. 
\end{proposition}
\begin{proposition}\label{prop:arbitragecalls}
Suppose $\set=\Omega$ and $\XX=\XX_c$. Then Condition \ref{cond:market set-up} (i),(ii) and (iv) are necessary and sufficient for absence of robust uniformly strong arbitrage. In consequence, when these conditions hold but Condition \ref{cond:market set-up} (iii) fails there is no robust uniformly strong arbitrage but a market calibrated model does not exist.
\end{proposition}
We defer proofs of Propositions~\ref{proposition:no_arbitrage} and \ref{prop:arbitragecalls} to Sections~\ref{Appendix:B} and \ref{Appendix:B_continued}.

\begin{remark}
If we assume that there is no robust uniformly strong arbitrage, then we can immediately deduce that $c_i(0)$ is non-increasing in $i$. Indeed, if there exists some $i$ such that $c_i(0)<c_{i+1}(0)$ then by taking $X=(\S_{i}-0)^+-(\S_{i+1}-0)^+$, $\Delta_i=1$ and $\Delta_j=0$ for $j\neq i$, we have $\PP(X)= c_i(0)-c_{i+1}(0)<0$ but $\Psi_{X,\Delta}=0$ which shows that $(X,\Delta)$ is a robust uniformly strong arbitrage.
\end{remark}

\subsection{Robust pricing--hedging duality and (super-)martingale optimal transport}
Our main theorem in the section states that the pricing-hedging duality is preserved under no short-selling restrictions when call options are traded. 
\begin{theorem}\label{thm:bubble_duality}
Suppose the market input $(\XX_c,\PP,\set)$ admits no WFLVR. Let $G:\R_+^{n}\to [-\infty,\infty)$ be an upper semi-continuous function such that 
\begin{equation}\label{eq:Gassumption}
G(s_1,\ldots,s_n)\le K(1+s_1+\ldots+s_n)
\end{equation}
on $\R_+^n$ for some constant $K$. Then the pricing--hedging duality holds, i.e.
\begin{align}\label{eq:thm_bubble_duality_original}
\primalpc(G)=\dualpc(G).
\end{align}
\end{theorem}
\begin{remark}
Our proof of this result follows closely \citet{mass_transport} and is an application of the duality theory from optimal transport, which allows us to express the dual problem as a min-max calculus of variations where the infimum is taken over functions corresponding to the delta hedging terms and marginal constraint, and the supremum is taken over all calibrated market models. The proof will be given in Sections~\ref{subsection:proof of theorem} and \ref{subsection:proof of theorem_continued}.
\end{remark}
\begin{remark}\label{rk:callsbubble}
Recall from Remark \ref{rk:primallessthandual} that the case of strict inequality in \eqref{eq:primallessthandual} may be thought of as a natural model for a financial bubble. From \eqref{eq:thm_bubble_duality_original} we see that this never happens when call options are traded, $\XX=\XX_c$. It is still possible that 
$$s_0>c_n(0)=\primalpc(\S_n)=\dualpc(\S_n)$$
so that the \emph{market price} for the asset $S$, which is $s_0$ is strictly greater than its fundamental price $c_n(0)$. However it is not clear if this could be seen as a bubble. In this case the market does not satisfy the \emph{no dominance} principle of \citet{merton_no_dominance1973}: the asset $\S$ is strictly dominated by a call with zero strike and one could argue that $c_n(0)$ is in fact the correct market price for $\S$. This situation is akin to the case of bubbles in complete markets described in \citet{JPS}. We will see in Section \ref{sec:puts_noduality_bubble} below that bubbles appear in a meaningful way when put options and not call options are traded.
\end{remark}

Note that, by Proposition \ref{proposition:no_arbitrage}, absence of WFLVR is equivalent to $\MM^-_{\XX_c,\PP,\set}\neq\emptyset$ and it implies Condition \ref{cond:market set-up}. Following classical arguments, going back to \citet{breeden1978prices}, we can then define probability measures $\mu_i$ on $\R_+$ by
\begin{align}\label{eq:def_mus}
\mu_i([0,K])= 1+c^{\prime}_{i}(K) \quad \text{ for }K\in \R_+.
\end{align}
Naturally the market prices $\PP$, or $c_i(K)$, are uniquely encoded via $(\mu_i)$ with $c_i(K)=\PP((\S_i-K)^+)=\int (s-K)^+\mu_i(\td s)$. To make the link with the (super-)martingale transport explicit, we may think of $(\mu_i)$ as the inputs. Note that by Remark \ref{rk:marketmodels} the set of calibrated market models $\MM^-_{\XX_c,\PP,\set}$ is simply the set of probability measures $\P$ on $\R_+^n$ such that $\S_0=s_0$, $\S$ is a supermartingale and $\S_i$ is distributed according to $\mu_i$. Accordingly we use the notation $\MM^-_{\XX_c,\PP,\set}=\MMUS$ and $\primals(G):=\sup_{\P\in \MMUS}\E_{\P}[G]$. Likewise we write $\dualp(G)=\duals(G)$. Note that we have dropped the explicit reference to call options. This is justified since in fact we can allow any $\mu_i$-integrable functions for the static part of trading strategies. To state this as a corollary we first rewrite Condition \ref{cond:market set-up} in terms of $\mu_1,\ldots,\mu_n$ as follows:
\begin{assumption}
\label{assumption:bubble_measures} 
The probability measures $\mu_1,\ldots,\mu_n$ on $\R_+$ satisfy
\begin{enumerate}
\item $s_0\ge\int_{\R_+} x\mu_1(dx)\ge \ldots\ge\int_{\R_+} x\mu_n(dx)$;
\item\label{item:bubble_measures_2} for any concave non-decreasing function $\phi:\R\to \R_+$, the sequence $(\int\phi \, d\mu_i)_{1\leq i\leq n}$ is non-increasing. 
\end{enumerate}
\end{assumption}
Then Theorem \ref{thm:bubble_duality} may be restated as follows.
\begin{corollary}\label{cor:bubble_duality}
Assume that $\mu_1,\ldots,\mu_n$ satisfy Assumption~\ref{assumption:bubble_measures} and that $\MMUS\neq \emptyset$. Let $G:\R_+^{n}\to [-\infty,\infty)$ be an upper semi-continuous function which satisfies \eqref{eq:Gassumption}. Then   
\begin{equation}\label{eq:cor_bubble_duality}
\begin{split}
\primals(G)
=&\inf\Big\{\sum_{i=0}^n\int u_i(s)\mu_i(\td s)\,:\, u_i:\R_+\to \R \textrm{ with linear growth and}\\
&\hspace{2.5cm}\Delta_i:\R_+^i\to \R_+ \textrm{ bounded s.t. }\Psi_{(u_i),(\Delta_i)}\ge G\,\text{ on }\set\Big\}\\
=&\duals(G),
\end{split}
\end{equation}
where $\mu_0:= \delta_{s_0}$. Further, if $\int x\mu_i(\td x)=s_0$ for $i=1,\ldots, n$ then \eqref{eq:cor_bubble_duality} also holds with $\Delta_i:\R_+^i \to \R$.
\end{corollary}
\begin{proof}
By taking expectations in the pathwise superhedging inequality, for any $\P\in \MMUS$ and any superhedging strategy $(u_i), (\Delta_i)$ as in \eqref{eq:cor_bubble_duality} we have 
\begin{align*}
\E_{\P}[G]\leq \sum_{i=1}^n\int u_i(s)\mu_i(\td s)
\end{align*}
and hence an inequality ``$\leq$" follows in the first equality in \eqref{eq:cor_bubble_duality}, see also Remark \ref{rk:primallessthandual}. An inequality ``$\leq$" in the second equality in \eqref{eq:cor_bubble_duality} is obvious because we take $\inf$ over a smaller set of superhedging strategies. The statement then follows by Theorem \ref{thm:bubble_duality}. The last statement is clear since in the special case that $\mu_1,\ldots,\mu_n$ have the same mean, $\MMUS$, $\S$ is the set of martingale measures with marginals $(\mu_i)$.
\end{proof}
\begin{remark}
\eqref{eq:cor_bubble_duality} is a general statement of Corollary 1.1 in \citet{mass_transport} in the presence of a prediction set $\set$, while $\primals(G)=\duals(G)$ follows directly from Theorem~\ref{thm:bubble_duality}.

The implication of Corollary~\ref{cor:bubble_duality} is that in a market without bubbles, a short-selling ban does not make any difference to the robust superhedging prices, and the martingale transport cost of $G$ with prediction set $\set$ is equal to the robust ($\set$)--superhedging price of $G$.
\end{remark}

\section{Put options as hedging instruments}\label{sec:dt_puts}
We specify now to the case when put options are traded, $\XX=\XX_p$. The set of semi-static trading strategies $(X,\Delta)$ is denoted $\AA_p$. In this case the options can not be used to super-replicate the asset. This, as we shall see, has important consequences for pricing and hedging.

\subsection{Pricing--hedging duality for options with bounded payoffs}
We start with a brief discussion of the market input, no arbitrage and existence of calibrated market models.
\begin{condition}\label{cond:market set-up put}
Let $\XX=\XX_p$ and $p_i(K):=\PP((K-\S_i)^+)$, $i=1,\ldots, n$, $K\geq 0$. Then
\begin{enumerate}
\item[(i)] \label{item:market set-up put 1} $p_i(x)$ is a non-negative, convex, increasing function of $x$ on $\R_+$, 
\item[(ii)] \label{item:market set-up put 2}$s_0\ge \lim_{x\to \infty}(x-p_1(x))\ge\ldots\ge \lim_{x\to \infty}(x-p_n(x))\ge 0$ and $0\le p^{\prime}_{i}(0+)\le 1$, 
\item[(iii)] \label{item:market set-up put 3}$p_i(K)\to 0$ as $K\to 0$, 
\item[(iv)] \label{item:market set-up put 4}for any $x\in\R_+$, $p_i(x)$ is non-decreasing in $i$.
\end{enumerate}
\
\end{condition}
An analogous statement of a robust Fundamental Theorem of Asset Pricing to the one in Proposition \ref{proposition:no_arbitrage} holds also in this setup.
\begin{proposition}\label{proposition:no_arbitrage put}
Suppose $\set$ is a closed subset of $\Omega$ and $\XX=\XX_p$. Then there is no WFLVR if and
only if there exists a market calibrated model, which then implies 
Condition \ref{cond:market set-up put}. Furthermore, if $\set=\Omega$ then Condition \ref{cond:market set-up put} implies absence of WFLVR. 
\end{proposition}
A direct analogue of Proposition \ref{prop:arbitragecalls} holds also in this setup. Further, 
if Condition \ref{cond:market set-up put} is satisfied, similarly to \eqref{eq:def_mus}, 
we can define probability measures $\mu_i$ on $\R_+$ by
\begin{align}\label{eq:def_mus_put}
\mu_i([0,K])= p_{i}^{\prime}(K) \quad \text{ for }K\in \R_+,
\end{align}
which will satisfy the same properties as before, namely Assumption \ref{assumption:bubble_measures}. The set of calibrated market models is simply $\MM^-_{\XX_p,\PP,\set}=\MMUS$ and only depends on the marginals $(\mu_i)$ and not on whether these were derived from put or from call prices. In consequence we have $\primalpp(G)=\primals(G)$.

The situation on the dual side ---  the superhedging problem --- is different. Indeed, we saw in Corollary \ref{cor:bubble_duality}, that in the case of call options we could relax the static part of the portfolio from combinations of call options to combinations of any functions with linear growth without affecting the superhedging price. In contrast, when put options are traded, their combinations are always bounded and such a relaxation may not be possible. We stress this in the notation and write $\duals^{(p)}(G):=\dualpp(G)$. Our first result shows that when $G$ is bounded then trading puts instead of calls has no impact on the superheding price, as one would expect.

\begin{theorem}\label{thm:bubble_duality_put}
Suppose the market input $(\XX_p,\PP,\set)$ admits no WFLVR or equivalently that $\MMUS\neq \emptyset$. In particular, Condition \ref{cond:market set-up put} is satisfied and \eqref{eq:def_mus_put} defines measures which satisfy Assumption~\ref{assumption:bubble_measures}. Let $G:\R_+^{n}\to [-\infty,\infty)$ be an upper semi-continuous function bounded from above. Then 
\begin{align}\label{eq:thm_bubble_duality_put}
\primals(G)=\primalpp(G)=\dualpp(G)=\duals^{(p)}(G).
\end{align}
\end{theorem}
The proof will be given in Sections~\ref{subsection:proof of theorem} and \ref{subsection:proof of theorem_continued} and is similar to the proof of Theorem~\ref{thm:bubble_duality}.

The above result may be extended to functions $G$ which are not necessarily bounded but have sub-linear growth. We state one such extension which will be used later. In contrast, the duality in \eqref{eq:thm_bubble_duality_put} will fail for $G$ which has a linear growth -- a theme we explore in the subsequent sections.
\begin{corollary}\label{eq:puts_noduality_extended}
In the setup of Theorem \ref{thm:bubble_duality_put}, assume that $G$ is an upper semi-continuous function such that for any $M>1$, $G_M(s_1,\ldots,s_n):=G(s_1,\ldots,s_n)-(s_1+\ldots+s_n)/M$ is bounded from above. Then \eqref{eq:thm_bubble_duality_put} holds for $G$.
\end{corollary}
\begin{proof}
We have
\begin{align*}
\duals^{(p)}(G)\le& \duals^{(p)}(G_M)+\duals^{(p)}((\S_1+\cdots+\S_n)/M)\\
\le& \primals(G_M)+\duals^{(p)}((\S_1+\cdots+\S_n)/M)\\
\le& \primals(G)+\frac{ns_0}{M},
\end{align*}
where we used the obvious inequality $0\le\duals^{(p)}(\S_i)\le s_0$, $i=1,\ldots,n$. Letting $M \to\infty$ we obtain $\duals^{(p)}(G)\le \primals(G)$.
The other inequality $\duals^{(p)}(G)\ge \primals(G)$ is true in all generality, (see Remark \ref{rk:primallessthandual}) and we conclude that $\duals^{(p)}(G)=\primals(G)$. 
\end{proof}

\subsection{Duality gap and bubbles}
\label{sec:puts_noduality_bubble}
We come back to the topic of financial bubbles considered in Remarks \ref{rk:primallessthandual} and \ref{rk:callsbubble}.  We start with a motivating example of a simple one period model, $n=1$. The prediction set is of the form $\set=\{s_0\}\times \set_1$ for some $\set_1\subset \R_+$. We assume the market admits no WFLVR which is equivalent to saying that $\mu$ defined via \eqref{eq:def_mus_put} is a probability measure supported on $\set_1$ and satisfies $\int x\mu(\td x)\leq s_0$. We assume the prediction set $\set_1$ is unbounded and consider an option with an upper-semi-continuous payoff function $G:\R_+\to [-\infty,\infty)$ such that $|G(x)|\le K|x|$ for some $K$ and let
\begin{align*}
\limsup_{x\to\infty,\,x\in \set_1}\frac{G(x)}{x}=:\beta\in[-\infty,\infty).
\end{align*}
A semi-static trading strategy is a pair $(X,\Delta)\in\AA_p$, $X\in \XX_p$ and $\Delta\geq 0$. If it super-replicates $G$
\begin{align*}
\Psi_{X,\Delta}(s_1):=X(s_1)+\Delta(s_{1}-s_0)\ge G(s_1), \quad s_1\in \set_1,
\end{align*}
then necessarily $\Delta\ge \beta_+$ since $X$ is bounded, where $\beta_+=\beta\vee 0$. Therefore we find 
\begin{align}\label{eq:dt_dualitygap}
V_{\mu,\set}^{(p)}(G)=&\inf\Big\{\int X(s_1)\mu(\td s_1):  (X,\Delta)\in\AA_p \text{ s.t }\Psi_{X,\Delta}\ge G\,\text{ on }\set\Big\} \nonumber\\
=&\inf_{\Delta_0\ge \beta_+}\bigg\{\Delta_0 s_0+\inf\Big\{\int X(s_1)\mu(\td s_1):  (X,\Delta)\in\AA_p  \text{ s.t }\nonumber\\
&\hspace{6cm} \Psi_{X,\Delta}(s_1)\ge G(s_1)-\Delta_0s_1\,\forall\,s_1\in \set_1\Big\}\bigg\}\nonumber\\
=&\inf_{\Delta_0\ge \beta_+}\Big\{\Delta_0 s_0+P_{\mu,\set}(G(\S)-\Delta_0\S_1)\Big\}\nonumber\\
=&\inf_{\Delta_0\ge \beta_+}\Big(\Delta_0 s_0+\int_{\R_+}(G(s_1)-\Delta_0 s_1)\mu(\td s_1)\Big)\\
=&\int_{\R_+}G(s_1)\mu(\td s_1)+\inf_{\Delta_0\ge \beta_+}\Big(\Delta_0(s_0-\int_{\R_+}s_1\mu(\td s_1))\Big)\nonumber\\
=&\int_{\R_+}G(s_1)\mu(\td s_1)+\beta_+(s_0-\int_{\R_+}s_1\mu(\td s_1))=P_{\mu,\set}(G)+\beta_+(s_0-\int_{\R_+}s_1\mu(\td s_1)).\nonumber
\end{align}
It follows that if the mean of $\mu$ is strictly smaller than $s_0$ then we have a duality gap for $G$ with linear growth. The intuitive reason is clear: buying the asset directly is implicitly more expensive than constructing a position using put options. If $G$ has bounded payoff then the latter is feasible as seen in Theorem \ref{thm:bubble_duality_put}. However for $G$ with a linear growth any superhedging portfolio has to include the asset $\S$ and is hence more expensive, as seen above.
When $G(s_1)=(s_1-K)^+$, $\limsup_{x\to \infty,\,x\in \set_1}G(x)/x=1$ and we obtain
\begin{align*}
V^{(p)}_{\mu,\set}(G)=P_{\mu,\set}(G)+\left(s_0-\int_{\R_+}x\mu(\td x)\right).
\end{align*}
Likewise, taking $G(s_1)=s_1$, we have 
$$s_0=V^{(p)}_{\mu}(\S_1)\geq P_{\mu,\set}(\S_1)=f_0:=\int_{\R_+}x\mu(\td x).$$
The market has a bubble --- a misalignment of market and fundamental prices --- if the forward price $f_0$ implied by the put options is strictly smaller than the spot price $s_0$. This should be contrasted with the situation in Remark \ref{rk:callsbubble}, where the bubble arose due to dominated assets.

The difference between these situations can be summarised as follows: in order to have a financially meaningful market, we must always have the following inequalities:
\begin{equation*}
  \text{\parbox[c]{2cm}{\centering market price}} \ \ge\  \parbox{3.5cm}{\parbox[c][][c]{3.5cm}{\centering cheapest}\\ \parbox[c][][c]{3.5cm}{\centering super-replication} \\  \parbox[c][][c]{3.5cm}{\centering price}} \ \ge \ \sup\left\{ \text{model-implied prices}\right\} \ = \ \text{\parbox[c]{3cm}{\centering fundamental price}}
\end{equation*}
The first inequality here follows from the fact that we can super-replicate an asset by purchasing it, and we may have a strict inequality without a simple arbitrage if it is not possible (due to portfolio constraints) to short-sell the asset (indeed, in this paper, we always interpret the market price of non-traded assets as their super-replication price; for traded assets, this is a simple consequence of the main results). However, in the case where there is a strict inequality here, the market contains a dominating portfolio --- that is, the super-replicating strategy strictly dominates the purchase of the asset at the market price, and so Merton's no-dominance principle fails. In general, one would not expect such markets to exist --- even if arbitrage were not possible, one would expect equilibrium to close the gap, since no (rational) market participants would purchase the asset at its market price. On the other hand, the second inequality here is rational --- there is no {\it a priori} need for the super-replication price and the model-implied prices to agree.

As a result, markets where a genuine difference between a fundamental and market price as defined above are at least mathematically tractable. One of the main contributions of this paper is that we provide a specific characterisation of markets where this is possible. In a more classical framework, the two cases described above are encapsulated in the difference between the complete setting of (for example) \citet{CH} and \citet{JPS}, where completeness of the markets mean that there is always equality between the cheapest super-replicating strategy and the model-implied price of an option, and the incomplete models of \citet{JPS2}, where Merton's no-dominance condition enforces the first inequality. However, in \citet{JPS2}, the existence of a bubble depends on the choice of some pricing measure to determine the `market-price'. In the current (robust) setting, we are able to define the market price in a concrete manner, leading to a possibly clearer characterisation of a bubble, which does not need some external `selection' procedure.

We now extend the above discussion to a general $n$-marginal setting. In the one-period case above, any $G$ was a European option and the size of the gap between its market and fundamental prices was given simply as a product of its linear growth coefficient and the bubble size $s_0-f_0$. In the general setting we can not compute explicitly the duality gap for an arbitrary payoff $G$. We give below a characterisation which then allows us to obtain explicit expressions for most of the typically traded exotic options.

\begin{theorem}\label{thm:no duality put}
  Assume $\mu_1,\ldots,\mu_n$ satisfy Assumption~\ref{assumption:bubble_measures}. Suppose the payoff function $G:\R_+^n\to [-\infty,\infty)$ is upper semi-continuous and
\begin{equation}\label{eq:case study 3 linear growth}
G(s_1,\ldots,s_n)\le K(1+s_1+\ldots+s_n)
\end{equation}
on $\R_+^n$ for some $K$. Define $\beta_i:\R_+^{i}\to \R$ recursively by setting $\beta_n=0$ and 
\begin{align}
\beta_i(s_1,\ldots,s_i)=&\sup_{s_{i+2},\ldots,s_n\in \R_+}\,\limsup_{x \to \infty}\bigg(\Big(\beta_{i+1}(s_1,\ldots,s_i,x)\nonumber\\
&+\frac{G(s_1,\ldots,s_i,x,s_{i+2},\ldots,s_n)}{x}\Big)\indicators{\set}(s_1,\ldots,s_i,x,s_{i+2},\ldots,s_n)\bigg)\vee 0,\label{eq:formula_beta}
\end{align}
for $i=0,\ldots,n-1$. If $G_{\beta}(\S):=G(\S)-\beta_0(\S_{1}-s_0)-\sum_{i=1}^{n-1}\beta_i(\S_1,\ldots,\S_{i})(\S_{i+1}-\S_i)$ is upper semi-continuous and bounded from above then 
\begin{align*}
V^{(p)}_{\vec{\mu},\set}(G)
=\sup_{\P\in\MMUS}\E_{\P}\Big[G-\beta_0(\S_{1}-s_0)-\sum_{i=1}^{n-1}\beta_i(\S_1,\ldots,\S_{i})(\S_{i+1}-\S_i)\Big]=P_{\vec{\mu},\set}(G_{\beta}).
\end{align*}
More generally, the result remains true if there exists a sequence $(0,\beta^{(N)})\in\AA_p$ such that $G_{\beta^{(N)}}(\S)$ is upper semi-continuous, bounded from above on $\set$ for every $N$ and $G_{\beta^{(N)}}(\S)\to G_{\beta}(\S)$ point-wise as $N\to \infty$.

\end{theorem}
The proof is reported in Section \ref{sec:proof_nmarg_duality}. Here we show how the above result applies in the case of an Asian or a Lookback option when $\set=\Omega$.

\begin{remark}
  This result stands in stark contrast with the existing literature on pricing under short-selling constraints: for example, in a general (classical) setting, where prices are assumed to be locally bounded semimartingales under some probability measure $\P$, and under a restriction on short-selling, \citet[Theorem~4.1]{Pulido} shows that there is no duality gap.
\end{remark}

\textit{Example 1: Asian option}

An Asian option has payoff function $G:\R_+^n\to \R_+$ defined by $$G(s_1,\ldots,s_n)=\Big(\frac{\sum_{i=1}^ns_i}{n}-K\Big)^+.$$
In this case, as for any $i=1,\ldots,n$ and $s_1,\ldots,s_i,s_{i+2},\ldots,s_n\in\R_+$
\begin{align*}
\lim_{x\to \infty}\frac{G(s_1,\ldots,s_i,x,s_{i+2},\ldots,s_n)}{x}=\frac{1}{n},
\end{align*}
\eqref{eq:formula_beta} can be simplified to 
\begin{align*}
\beta_i(s_1,\ldots,s_i)=\sup_{s_{i+2},\ldots,s_n\in \R_+}\,\limsup_{x \to \infty}\beta_{i+1}(s_1,\ldots,s_i,x)+\frac{1}{n}.
\end{align*}
This yields
$\beta_i=(n-i)/n$ for $i=0,\ldots,n-1$. It is clear that
\begin{align*}
G_{\beta}(\S)=G(\S)-\beta_0(\S_{1}-s_0)-\sum_{i=1}^{n-1}\beta_i(\S_{i+1}-\S_i)= \Big(\frac{\sum_{i=1}^n\S_i}{n}-K\Big)^+-\frac{\sum_{i=1}^n\S_i}{n}+s_0
\end{align*} 
is continuous and bounded from above. Therefore, by Theorem \ref{thm:no duality put}
\begin{align*}
V^{(p)}_{\vec{\mu}}(G)&=\sup_{\P\in\MMNS}\E_{\P}\Big[G-\beta_0(\S_{1}-s_0)-\sum_{i=1}^{n-1}\beta_i(\S_{i+1}-\S_i)\Big]\\
&=\sup_{\P\in\MMNS}\E_{\P}[G]+\frac{1}{n}\sum_{i=1}^n\Big(s_0-\int_{\R_+}x\mu_i(dx)\Big).
\end{align*}

\textit{Example 2: Lookback option with knock-in feature}

The second example we consider is a Lookback option with a knock-in feature, whose payoff function $G:\R_+^n\to \R_+$ is given by $$G(s_1,\ldots,s_n)=\big(\max_{0\le i\le n}s_i-K\big)^+\indicators{\displaystyle \min_{0\le i\le n}s_i\le B}.$$
In particular, when $B=\infty$, it is just a lookback call option with strike $K$. \\
By \eqref{eq:formula_beta} 
$$\beta_{n-1}(s_1,\ldots,s_{n-1})=\sup_{s_n\in\R_+}\limsup_{x\to \infty}\frac{G(s_1,\ldots,s_{n-1},x)}{x}=\indicators{\displaystyle \min_{0\le i\le n-1}s_i\le B}.$$
Since for $i=0,\ldots,n-2$ and $s_1,\ldots,s_i\in\R_+$,
\begin{align*}
\alpha_i(s_1,\ldots,s_i):=\sup_{s_{i+2},\ldots,s_n\in\R_+}\lim_{x \to \infty}\frac{G(s_1,\ldots,s_i,x,s_{i+2},\ldots,s_n)}{x}=1.
\end{align*} 
\eqref{eq:formula_beta} can then be simplified to 
\begin{align*}
\beta_i(s_1,\ldots,s_i)=\sup_{s_{i+2},\ldots,s_n\in \R_+}\,\limsup_{x \to \infty}\beta_{i+1}(s_1,\ldots,s_i,x)+1, \quad i=0,\ldots,n-2,
\end{align*}
from which we can derive 
$$\beta_i(s_1,\ldots,s_i)=(n-1-i)+\indicators{\displaystyle \min_{0\le j\le i}s_j\le B}, \quad i=0,\ldots,n-1.$$
It is not hard to see that 
\begin{align*}
G_{\beta}(\S)=&G(\S)-\beta_0(\S_{1}-s_0)-\sum_{i=1}^{n-1}\beta_i(\S_1,\ldots,\S_i)(\S_{i+1}-\S_i)\\
=&G(\S)+\beta_0s_0-\sum_{i=0}^{n-2}\big(\beta_i(\S_1,\ldots,\S_i) - \beta_{i+1}(\S_1,\ldots,\S_{i+1})\big)\S_{i+1}-\beta_{n-1}(\S_1,\ldots,\S_{n-1})\S_{n}\\
= &\Big(\max_{0\le i\le n}\S_i-K\Big)^+\indicators{\displaystyle \min_{0\le i\le n}\S_i\le B}-\sum_{i=1}^n\S_i+\sum_{i=1}^n\indicators{\min_{0\le j\le i-1}\S_j>B}\indicators{\S_i\le B} \S_i+ns_0
\end{align*}
is bounded from above. Now define continuous functions $\beta^{(N)}_i:\R_+^i\to \R_+$ by
\begin{align*}
\beta^{(N)}_i(s_1,\ldots,s_i)=\begin{cases}
n-i &\mbox{if } \displaystyle \min_{0\le j\le i}s_j\le B\\
n-i-1+  N\Big(B+\frac{1}{N}-\min_{0\le j\le i}s_j\Big)^+ &\mbox{otherwise. } 
\end{cases}
\end{align*}
Similar to above we can show that $G_{\beta^{(N)}}(\S)=G(\S)-\beta^{(N)}_0(\S_{1}-s_0)-\sum_{i=1}^{n-1}\beta^{(N)}_i(\S_1,\ldots,\S_{i})(\S_{i+1}-\S_i)$ is bounded from above. Also $\beta^{(N)}_i\to \beta_i$ as $N\to \infty$ for any $i=0,\ldots,n-1$. Then $G_{\beta^{(N)}}\to G_{\beta}$ pointwise as $N\to \infty$ and hence by Theorem \ref{thm:no duality put} 
\begin{align*}
V^{(p)}_{\vec{\mu}}(G)&=\sup_{\P\in\MMNS}\E_{\P}\Big[G-\beta_0(\S_{1}-s_0)-\sum_{i=1}^{n-1}\beta_i(\S_1,\ldots,\S_i)(\S_{i+1}-\S_i)\Big].\\
\end{align*}

As shown in the Lookback option example above, the duality gap is not only dependent on $G$ and the marginal distributions $(\mu_i)$, but also on how the $\mu_i$'s are (optimally) transported. In the case that $\set$ is a strict subset of $\Omega$, it may become increasingly hard to calculate $\beta$ and check the assumption of Theorem \ref{thm:no duality put}. We develop now an argument which connects asymptotically the duality gap of $G$ in the presence of prediction set and the duality gaps of penalised functions of $G$ in the absence of a prediction set. In particular, it provides an alternative way to compute the duality gap when $\set$ is an arbitrary closed set.

Assume the market input $(\XX_p,\PP,\set)$ admits no WFLVR and $G$ is upper-semi continuous subject to
\begin{equation*}
G(s_1,\ldots,s_n)\le K(1+s_1+\ldots+s_n),\quad (s_1,\ldots,s_n)\in \R_+.
\end{equation*} 
Under this assumption, we are going to argue first if $\primals(G)=-\infty$, then $\duals^{(p)}(G)=-\infty$. By Proposition \ref{proposition:no_arbitrage put}, absence of WFLVR is equivalent to $\MM^-_{\XX_p,\PP,\set}\neq\emptyset$. It follows from the sublinearity of $\duals^{(p)}$ and Theorem \ref{thm:bubble_duality_put} that
\begin{align*}
\duals^{(p)}(G)\le&\, \duals^{(p)}\big(G-K(1+\S_1+\cdots+\S_n)\big)+ \duals^{(p)}\big(K(1+\S_1+\cdots+\S_n)\big)\\
\le&\, \primals(G-K(1+\S_1+\cdots+\S_n))+n Ks_0+K=-\infty.
\end{align*}

From now on we make an additional assumption that $\primals(G)>-\infty$. 
Define $G^{(N)}:\R_+^n\to [-\infty,\infty)$ by
\begin{align*}
G^{(N)}(s_1,\ldots,s_n)=G(s_1,\ldots,s_n)-N\lambda_{\set}(s_1,\ldots,s_n),
\end{align*}
where $\lambda_{\set}(\S):=(1+\S_1+\ldots+\S_n)\indicator{(\S_1,\ldots,\S_n)\notin \set}$ is as defined in $\eqref{eq:penalty_prediction_set}$. Then note that
\begin{align}
\duals^{(p)}(G)= \inf_{N\ge 1}V^{(p)}_{\vec{\mu}}(G^{(N)}).\label{eq:put_dual_asymptotic}
\end{align}
The inequality ``$\leq$" is clear. On the other hand, ``$\geq$" follows from the fact that $G^{(N)}$ is decreasing in $N$ and given any $(X,\Delta)\in\AA_p$, $\Psi_{X,\Delta}\ge -N(1+\S_1+\ldots+\S_n)$ for $N$ sufficiently large.

Since $\set$ is closed, $-\indicator{(\S_1,\ldots,\S_n)\notin \set}$ is an upper semi-continuous function and hence $G^{(N)}$ is upper semi-continuous. Then the problem is reduced to the case that $\set=\Omega$, for which we have a formula to calculate the duality gap if the contingent claim satisfies all the assumptions in Theorem \ref{thm:no duality put}. Now let 
$$\gamma_N:=V^{(p)}_{\vec{\mu}}(G^{(N)})-\sup_{\P\in \MMNS}\E_{\P}[G^{(N)}].$$
It follows by $\eqref{eq:put_dual_asymptotic}$ that 
$$\duals^{(p)}(G)= \inf_{N\ge 1}\Big\{\sup_{\P\in \MMNS}\E_{\P}[G^{(N)}]+\gamma_N\Big\}=\lim_{N\to \infty}\Big\{\sup_{\P\in \MMNS}\E_{\P}[G^{(N)}]+\gamma_N\Big\}.$$
In addition, we can deduce that
\begin{align*}
\inf_{N\ge 1}\sup_{\P\in \MMNS}\E_{\P}[G^{(N)}]= \sup_{\P\in \MMNS}\inf_{N\ge 1}\E_{\P}[G^{(N)}]=\sup_{\P\in\MMUS}\E_{\P}[G],
\end{align*}
where the first equality is achieved by using the Min-Max Theorem (Corollary 2 in \citet{min_max_terkelsen1972}) and the second equality holds since $\inf_{N\ge 1}\E_{\P}[G-N\lambda_{\set}]=-\infty$ for any $\P\in (\MMNS\setminus\MMUS)$ but
$$\inf_{N\ge 1}\sup_{\P\in \MMNS}\E_{\P}[G^{(N)}]\ge\sup_{\P\in\MMUS}\E_{\P}[G]>-\infty.$$
Hence, the limit of $\gamma_N$ exists and by writing $\gamma=\lim_{N\to \infty}\gamma_N$ we have
\begin{align*}
\duals^{(p)}(G)=\sup_{\P\in\MMUS}\E_{\P}[G]+\gamma.
\end{align*}

\section{Continuous time: local martingales, bubbles and pricing}\label{sec:continuous}
We turn now to continuous time models to explore the link between options' prices, trading constraints, speculative bubbles and strict local martingales. Let $\Omega = \D([0,T],\R_+)$ be the space of non-negative right-continuous functions with left limits on $[0,T]$ and $\S=(\S_t:t\leq T)$ be the canonical process on $\Omega$ with $(\FF_t)$ denoting its natural filtration. 

Consider the case when put options trade for $n\geq 1$ maturities $0<T_1<\ldots<T_n=T$,
$$\XX_p:= \{(K-\S_{T_i})^+: 1\leq i\leq n, K\geq 0\},\quad  p_i(K):= \PP((K-\S_{T_i})^+).$$
We need to impose some assumptions on the prediction set $\set$. 
\begin{assumption}\label{ass:set}
 The prediction set $\set\subset \Omega$ satisfies $\omega(0)=s_0$ for every $\omega\in \set$ and 
 $$\textrm{ for any $\omega\in \set$ and any stopping time }\tau,\quad \omega^\tau=(\omega(t\land \tau(\omega)):t\leq T)\in \set.$$
 Further the set $\set_{\vec{T}}:=\{(\omega_0,\omega_{T_1},\ldots, \omega_{T_n}): \omega \in  \set\}$ is closed.
\end{assumption}
The first condition corresponds to $\set$ being closed under stopping and will imply that any superhedging strategy in fact satisfies a collateral requirement, see Remark \ref{rk:collateral} below. The second point is technical and will enable us to compare the continuous time setting to the discrete time setting.

There are several possible choices for the class of admissible dynamic trading strategies. They will typically lead to the same superhedging price, provided the admissible class is large enough, but to different sets of calibrated market models. Here, to make the connection with discrete time setup clearer, we consider dynamic trading strategies $\Delta$ which are predictable piecewise constant processes with finitely many jumps. More precisely, $\Delta:[0,T]\times \Omega \to \R$ such that for any $\omega\in \Omega$, $\Delta(\omega):[0,T]\to \R_+$ is a simple non-negative function (piecewise constant with finitely many jumps) and for any $t\in [0,T]$ and for any $\omega_1,\omega_2\in \Omega$ such that $\omega_1(s)=\omega_2(s)$ for $s\in [0,t)$ we have $\Delta_t(\omega_1)=\Delta_t(\omega_2)$.
 We say that such $\Delta$ is admissible, and write $\Delta\in\AA$. Note that for $\Delta\in \AA$ the stochastic integral $\int_0^t \Delta_{u-}\td \S_u$ is a sum and hence is defined pathwise.  

An admissible semi-static trading strategy is a pair $(X,\Delta)$ with a linear combination of put options $X(\omega)=a_0+\sum_{i=1}^m a_i X_i(\omega)$, $m\geq 0$, $a_i\in \R$, $X_i\in \XX_p$ and $\Delta\in \AA$. Its payoff at time $T$ is given by
$$\Psi_{X,\Delta}(\S)=X(\S)+\int_0^T \Delta_{u-}\td \S_u.$$
Recall that the set of admissible semi-static trading strategies is denoted $\AA_p=\AA_{\XX_p}$ and the super-replication price $\dualpp$ was given in Definition  \ref{def:superreplication}.
\begin{remark}\label{rk:collateral}
Note that because $\set$ is closed under stopping (cf.\ Assumption \ref{ass:set}), it follows that if $(X,\Delta)\in \AA_p$ superhedges $G$ on $\set$ then in fact 
\begin{equation}\label{eq:collateral}
 \Psi_{X,\Delta}(\S^t)=X(\S^t)+ \int_0^t \Delta_{u-}\td \S_u\geq G(\S^t),\quad t\leq T, \textrm{ on }\set,
\end{equation}
where $\S^t=(\S_{u\wedge t}:u\leq T)$. 
In other words, $(X,\Delta)$ satisfies a collateral requirement. As we will see below, this feature will contribute towards emergence of bubbles.
\end{remark}

Static trading arguments, as in the proof of Propositions \ref{proposition:no_arbitrage}, show that absence of WFLVR implies
that $p_i(K)$ satisfy the properties listed in Condition \ref{cond:market set-up put} and hence we can use \eqref{eq:def_mus_put} to define probability measures $\vec{\mu}=(\mu_i)_{i=1}^n$ which satisfy Assumption \ref{assumption:bubble_measures}. 
The set of calibrated market models $\MM^-_{\XX,\PP,\set}$ is as given in Definition \ref{def:marketmodel}. 
Note that Remark \ref{rk:primallessthandual} is in force, with the convention $\infty-\infty=-\infty$. 
Finally, let $\lMMUS$ be the set of all calibrated local martingale measures on $(\Omega, \F_T)$, i.e.\  $\P$ such that $\S$ is a $\P$--local martingale and $\E_\P[(K-\S_{T_i})^+]=p_i(K)$, $K\geq 0$, or equivalently $\S_{T_i}\sim \mu_i$, $1\leq i\leq n$. It is easy to see that $\MM^-_{\XX,\PP,\set}$ is the set of measures $\P$ under which $\S_{T_i}\sim \mu$ and $\S$ is a $\P$-supermartingale and in particular $\lMMUS\subset \MM^-_{\XX,\PP,\set}$.
 
Consider now a European option with payoff $G(\omega)=G(\omega_{T_n})$, or more generally an upper semi-continuous $G(\omega)=G(\omega_{T_1},\ldots,\omega_{T_n})$. 
We can then compare the present setting to that of a discrete $n$-period model with traded put options at prices $\PP$, where short-selling is prohibited, and with a prediction set $\set_{\vec{T}}$, as considered in Section \ref{sec:dt_puts}.  Denote the corresponding primal and dual values $\primalppd$ and $\dualppd$. Note that the discrete superhedging problem naturally embeds into the continuous time one. A discrete time trading strategy corresponds to a non-negative $(\Delta_t)$ constant on every $[T_i,T_{i+1})$, $i=0,\ldots,n-1$, which is in $\AA$. On the primal side, for any $\P\in \lMMUS$ the vector $(\S_0,\S_{T_1},\ldots,\S_{T_n})$ is a calibrated discrete time market model so for $G$ as above, calibrated continuous time models are embedded in discrete ones. In summary, 
\begin{equation}
\primalpp(G)\leq \primalppd(G)\quad \textrm{ and }\quad \dualpp(G)\leq \dualppd(G).
\end{equation}
In some cases we can establish an equality in the first inequality. For example, when $\set=\{\omega\in \Omega: \omega(0)=s_0\}$ then any discrete time market model may be seen as a continuous time one with the asset being constant on any $[T_i,T_{i+1})$. We can then conclude that there is no duality gap in the continuous time setting from the results Theorem \ref{thm:bubble_duality_put} and Corollary \ref{eq:puts_noduality_extended} in discrete time.

However, our prime interest is in the case when the pricing-hedging duality fails. We can use the results of Section \ref{sec:puts_noduality_bubble} to understand the case of European options. 

\begin{proposition} \label{prop:bubble} Suppose $\set$ satisfies Assumption \ref{ass:set} and $\set_{T}:=\{\omega(T):\omega\in \set\}$ is unbounded, and that market prices are such that $\lMMUS\neq \emptyset$.  Consider $G:\R_+\to [-\infty,\infty)$ a continuous function with linear growth (i.e.~\eqref{eq:Gassumption} holds) such that the limit $\beta:=\lim_{s\to \infty, s\in \set_T}\frac{G(s)}{s}$ is well defined and non-negative. Then,
\begin{equation}
\begin{split}
\label{eq:Euro_bubble}
\dualpp(G)=& \sup_{\P\in \lMMUS}\lim_{n\to \infty}\E_{\P}[G_+(\S_{T\wedge\tau_n^\P})-G_{-}(\S_{T})]\\
=&\int G(s)\mu_n(\td s)+\beta\Big(s_0-\int_{\R_+}s\mu_n(\td s)\Big),
\end{split}
\end{equation}
where we implicitly set $G(\omega)=G(\omega(T))$, where $(\tau_n^\P)$ is a localising sequence for $\S$ under $\P$ and $G_+=G\lor 0$, $G_-=-(G\land 0)$.
\end{proposition}
We first give two remarks  before proving the above result.
\begin{remark}
  If the forward price implicit in the put options, \[f_0=\int s\mu_n(\td s)=\lim_{K\to \infty} (K-p_n(K)),\] is cheaper than the spot, $s_0>f_0$, then the market has a bubble. The \emph{market price}, defined as the superhedging price, is strictly greater than the \emph{fundamental price}: \[\sup_{\P\in \lMMUS}\E_\P[G(\S_{T})] = \int G(s)\mu_n(\td s).\] The correction is equal to $\beta^+(s_0-f_0)$. This is the same correction as exhibited in Theorem 5.2 in \citet{CH}, see also Section 6.1 in \citet{JPS2}. Therein, the bubble was driven by a collateral requirement and a strict local martingale property. While the former is a natural trading restriction the latter appears artificial. 
  In our robust framework  a bubble is triggered by trading restrictions and properties of market prices of options. The difference is that 
  we take market prices as given and adopt a robust framework. A bubble arises when these prices are misaligned with the asset price $s_0>f_0$ while an arbitrage does not arise because of the trading restrictions. In our setup the trading restrictions take the form of a short selling ban and, as highlighted in Remark \ref{rk:collateral} above, a collateral requirement.
 \end{remark}
\begin{remark}
  The assumption $\lMMUS\neq \emptyset$ is an implicit assumption on $\set$ and market prices. It is satisfied e.g.\ when $\set$ is equal to all paths, or all continuous paths, which start in $s_0$ and put prices $p_i(K)$ satisfy the properties listed in Condition \ref{cond:market set-up put}. The latter is equivalent to $\vec{\mu}=(\mu_i)_{i=1}^n$, defined via \eqref{eq:def_mus_put}, satisfying Assumption \ref{assumption:bubble_measures}.
\end{remark}
\begin{proof}[of Proposition \ref{prop:bubble}.]
As explained above, we can directly compare the continuous time setting with a discrete time setting from Section \ref{sec:dt_puts} with the same put prices and prediction set $\set_{\vec{T}}=\{(\omega_0,\omega_{T_1},\ldots, \omega_{T_n}): \omega \in \set\}$. Using \eqref{eq:dt_dualitygap}, which is a one-marginal result, we immediately have $$ \dualppd(G)\leq V^{p}_{\mu_n,\set_{T}}(G)= \int G(s)\mu_n(\td s)+\beta^+\Big(s_0-\int_{\R_+}s\mu_n(\td s)\Big)$$ and hence we conclude that
\begin{equation}\label{eq:upperbound_dual}
\dualpp(G)\leq \int G(s)\mu_n(\td s)+\beta^+\Big(s_0-\int_{\R_+}s\mu_n(\td s)\Big).
\end{equation}

Consider a superhedging strategy $(X,\Delta)$ and $\P\in \lMMUS$ with $(\tau_n)$ a reducing sequence for $\S$ under $\P$. If $G(\S_{T\wedge\tau_n})> 0$, it follows from \eqref{eq:collateral} and $G_-\ge 0$ that
\begin{align*}
X(\S_{T\wedge\tau_n})+\int_0^{T\wedge\tau_n}\Delta_{u-} \td \S_u\ge G_+(\S_{T\wedge\tau_n})-G_-(\S_{T}).
\end{align*}
Otherwise, $G(\S_{T\wedge\tau_n})\le 0$ and then 
\begin{align*}
X(\S_{T})+\int_0^{T}\Delta_{u-} \td \S_u\ge G(\S_{T})\ge G_+(\S_{T\wedge\tau_n})-G_-(\S_{T})
\end{align*}
Therefore
\begin{align*}
G_+(\S_{T\wedge\tau_n})-G_-(\S_{T})\le& X(\S_{T})+\big(X(\S_{T\wedge\tau_n})-X(\S_{T})\big)\indicator{G(\S_{T\wedge\tau_n})> 0}+\int_0^{T}\tilde\Delta_{u-} \td \S_u,
\end{align*}
where $\tilde\Delta_u = \Delta_u\indicator{u\leq \tau_n\wedge T} + \Delta_u\indicator{u > \tau_n\wedge T}\indicator{G(\S_{T\wedge\tau_n})\le0}$. We note that $\tilde\Delta\in \AA$ and hence the expectation of the integral is non-positive under $\P$. Further, $\tau_n\wedge T=T$ for $n$ large enough (which may depend on the path) and $X$ is bounded so we may apply dominated convergence theorem to conclude that 
\begin{align}\label{eq:lowerbound_dual}
\limsup_{n\to \infty}\E_{\P}[G_+(\S_{T\wedge\tau_n})-G_{-}(\S_{T})]\le \E_{\P}[X(\S_{T})] = \PP(X)
\end{align}
and hence the LHS is a lower bound on $\dualpp(G)$. Finally we compute the LHS. Note that for any $\epsilon>0$, $G(s)-(\beta-\epsilon)s$ is bounded from below on $\set_{T}$. It follows, applying Fatou's Lemma, and noting that $\P \in \lMMUS$ implies $\S_T$ and $\S_{T \wedge \tau_n}$ are almost surely in $\set_T$, that
\begin{equation*}
\begin{split}
\liminf_{n\to\infty} \E_\P[G_+(\S_{T\wedge\tau_n})]\geq & \liminf_{n\to\infty} \E_\P[G_+(\S_{T\wedge\tau_n})-(\beta-\epsilon)\S_{T\wedge \tau_n}]+(\beta-\epsilon)s_0\\
\geq& \E_\P[G_+(\S_T)-(\beta-\epsilon)\S_T]+(\beta-\epsilon)s_0\\=& \int G_+(s)\mu_n(\td s) + (\beta-\epsilon)\left(s_0-\int s\mu_n(\td s)\right). 
\end{split}
\end{equation*}
We conclude that the upper bound in \eqref{eq:upperbound_dual} coincides with the lower bound obtained by taking $\epsilon\searrow 0$ and inf over superhedging strategies in \eqref{eq:lowerbound_dual}, as required.
\end{proof}

The above statement may be extended to $G$ which depends on the values of the asset at the intermediate maturities, i.e.\ $G(\S)=G(\S_{T_1},\ldots,\S_{T_n})$ using Theorem \ref{thm:no duality put}. We do not pursue it here.

\section{Appendix}\label{sec:app}

\subsection{Preliminary Results}\label{sec:ap-prelim}
In this section and in Sections~\ref{subsection:proof of theorem} and \ref{subsection:proof of theorem_continued}, we assume that $\mu_1,\ldots,\mu_n$ are probability measures on $\R_+$ which have finite first moment. Let $\Pi_{\vec{\mu}}$ be the set of all Borel probability measures on $\Omega$ with marginals $\delta_{s_0}, \mu_1,\ldots,\mu_n$ and denote $\MMNS$ the set of probability measures $\P$ on $\Omega$ such that $\S$ is a supermartingale and $\S_i$ is distributed according to $\mu_i$. We also write $C_b(\R^j_+,\R_+)$ to denote the set of continuous, bounded and non-negative functions $f$  on $\R^j_+$ and $C_{c}(\R^j_+,\R_+)$ for the subset of continuous non-negative and compactly supported functions. 
\begin{lemma}\label{lemma:super_martingale_equivalent}
Let $\pi\in \Pi_{\vec{\mu}}$. Then the following are equivalent:
\begin{enumerate}
\item $\pi\in \MMNS$.
\item For $0\le j\le n-1$ and for every $\Delta\in C_{c}(\R^{j}_+,\R_+)$, we have
\begin{align}\label{eq:super_martingale_equivalent_2}
\int_{\Omega}\Delta(x_1,\ldots,x_j)(x_{j+1}-x_j)\td \pi(x_1,\ldots,x_n)\le 0
\end{align} 
\end{enumerate}
\end{lemma}

\begin{proof}[Proof of Lemma~\ref{lemma:super_martingale_equivalent}]
(1) asserts that whenever $A\subseteq\R_+^{j}$, $j=0,\ldots,n-1$, is Borel measurable, then 
\begin{align}\label{eq:super_martingale_equivalent_1}
\int_{\Omega_{s_0}}\indicators{A}(x_1,\ldots,x_j)(x_{j+1}-x_j)\td \pi(x_1,\ldots,x_n)\le 0.
\end{align}
To see $\eqref{eq:super_martingale_equivalent_2}\Rightarrow
\eqref{eq:super_martingale_equivalent_1}$, we fix any $j=0,\ldots,n-1$ and $\Delta\in C_{c}(\R^j_+,\R_+)$ and define simple functions $f_k:\R_{+}^{j}\to \R$ by 
$f_k=2^{-k}\lfloor2^k \Delta\rfloor$.
Then $0\le f_k\uparrow \Delta$ and it follows from the Dominated Convergence Theorem and $\eqref{eq:super_martingale_equivalent_1}$ that $\eqref{eq:super_martingale_equivalent_2}$ is satisfied.

To show $\eqref{eq:super_martingale_equivalent_1}\Rightarrow \eqref{eq:super_martingale_equivalent_2}$, first consider $A\in\R_+^{j}$ such that $A$ is open and bounded. Note that $\indicators{A}$ is lower semi-continuous and hence there exists a sequence $(f_k)_{k\ge 1}\in C_{c}(\R^j_+,\R_+)$ such that $0\le f_k\le \indicators{A}$ and $f_k\uparrow \indicators{A}$.  Therefore the Dominated Convergence Theorem implies that
\begin{equation*}
\int_{\Omega}\indicators{A}(x_1,\ldots,x_j)(x_{j+1}-x_j)\td \pi(x_1,\ldots,x_n)\le 0.
\end{equation*}
Now $A\in\R_+^{j}$ is an arbitrary open set. We have $A=\cup_{n\ge 1}A^{(n)}$ with $A^{(n)}:=A\cap\{S\in\R_+^j:\|S\|<n\}$ being open and bounded. Then by Dominated Convergence Theorem
\begin{align*}
\int_{\Omega}\indicators{A}(x_1,\ldots,x_j)(x_{j+1}-x_j)\td \pi(x_1,\ldots,x_n)\le 0.
\end{align*}

At last, if $A\subseteq \R_+^{j}$ is a Borel set, then by Corollary 3.12 in \citet{bruckner2008real}, for every $N>0$, there is an open set $A_{N}\subseteq \R_+^{j}$ such that $A\subseteq A_{N}$ with $\pi(A_{N})\le \pi(A)+\frac{1}{N}$. It follows that
\begin{align*}
&\int_{\Omega}\indicators{A}(x_{j+1}-x_j)\td \pi(x_1,\ldots,x_n)\\
=& \int_{\Omega}\indicators{A_N}(x_{j+1}-x_j)\td \pi(x_1,\ldots,x_n)-\int_{\Omega}\indicators{A_N\setminus A}(x_{j+1}-x_j)\td \pi(x_1,\ldots,x_n) \\
\le& \int_{\Omega}\indicators{A_N\setminus A}(x_1,\ldots,x_j)x_{j}\td \pi(x_1,\ldots,x_n)\\
=& \int_{\Omega}\indicators{A_N\setminus A}\indicator{x_{j}\ge \sqrt{N}}x_{j}\td \pi(x_1,\ldots,x_n)+\int_{\Omega}\indicators{A_N\setminus A}\indicator{x_{j}< \sqrt{N}}x_{j}\td \pi(x_1,\ldots,x_n)\\
\le& \int_{\R_+}\indicator{x_{j}\ge \sqrt{N}}x_{j}\mu_j(dx)+\frac{\sqrt{N}}{N}\quad\to 0\,\text{ as } N\to \infty. 
\end{align*}

\end{proof}

\begin{lemma}\label{lemma:compactness_super_martingale}
For a closed $\set\subseteq \Omega$ the set $\MMUS$ is compact in the weak topology.
\end{lemma}
\begin{proof}
Since $\MMUS$ is a subset of the compact set $\Pi_{\vec{\mu}}$ it suffices to prove that $\MMUS$ is a closed subset of $\Pi_{\vec{\mu}}$. By Lemma~\ref{lemma:super_martingale_equivalent}, 
\begin{align*}
\MMNS=\bigcap_{0\le j\le n-1}\bigcap_{\substack{\Delta:\R^{j}_+\to \R_+\\ \text{continous and bounded}}}\Big\{\int_{\R_+^n}\Delta(x_1,\ldots,x_j)(x_{j+1}-x_j)\td \pi(x_1,\ldots,x_n)\le 0\Big\}.
\end{align*}
Therefore, by Lemma 2.2 in \citet{mass_transport}, $\MMNS$ is a closed subset of $\Pi_{\vec{\mu}}$ in the weak topology. 

To show $\MMUS$ is a closed subset of $\MMNS$, we take any sequence $(\Q_n)\in \MMN$ such that $\Q_n(\set) = 1$ and $\Q_n\to \Q$ for some $\Q\in \MMN$ as $n\to \infty$. Then by weak convergence of measures, for $\set\subseteq \Omega$ closed, $\Q(\set)\ge \limsup_{n\to \infty}\Q_n(\set)=1$. It follows that $\MMU$ is a closed subset of $\MMN$ and hence is closed in $\Pi_{\vec{\mu}}$ in the weak topology.
\end{proof}

To prove Theorems \ref{thm:bubble_duality} and \ref{thm:bubble_duality_put}, we will use the following Monge-Kantorovich duality theorem, which is essentially proposition 2.1 in \citet{mass_transport}. The proposition is rewritten here to suit the notation and purpose of this paper. 
\begin{lemma}\label{lemma:Monge-Kantorovich}
For any $G$ that is upper semi-continuous and bounded from above we have 
\begin{align}
\sup_{\pi\in\Pi_{\vec{\mu}}}\E_{\pi}[G]=\inf\Big\{\PP(X): X\in\AA_o,\text{ s.t. }
X\ge \phi \text{ on } \Omega\Big\},\quad  \textrm{where } o\in \{c,p\}.\label{eq:super-replicating_put_call}
\end{align}
Further, the result remains true with $o=c$ for any upper-semi continuous $G$ that satisfies \ref{eq:Gassumption}.
\end{lemma}
The call option case is just proposition 2.1 in \citet{mass_transport}. The put option case follows from Equation (A.1) in the proof of proposition 2.1 in \citet{mass_transport}, which in our notation simply states that for any $G$ that is upper semi-continuous and bounded from above
\begin{align}
\sup_{\pi\in\Pi_{\vec{\mu}}}\E_{\pi}[G]=\inf\Big\{\sum_{i=0}^n\int_{\R_+}u_i\td \mu_i: u_i\in C_b(\R_+,\R) \text{ s.t. }
\sum_{i=1}^n u_i(\S_i)\ge G(\S) \text{ on $\Omega$}\Big\}.\label{eq:super-replicating_put}
\end{align}
Note that given any $f\in C_b(\R_+,\R)$, $\epsilon>0$ and $i=1,\ldots,n$, there is some $u:\R_+\to \R$ taking the form $a_0+\sum_{j=1}^m a_j(K_j-s_j)^+$ such that $u\ge f$ and $\int (u-f)\td \mu_i<\epsilon$. Therefore we may change the class of admissible functions in \eqref{eq:super-replicating_put} from $C_b(\R_+,\R)$ to $u:\R_+\to \R$ taking the form $a_0+\sum_{j=1}^m a_j(K_j-s)^+$.

\subsection{Proofs of FTAP in the setting with traded calls or puts: case where $\set=\Omega$}\label{Appendix:B}
The proofs of Propositions~\ref{proposition:no_arbitrage} and \ref{proposition:no_arbitrage put} are virtually identical so we only give the proof of Proposition \ref{proposition:no_arbitrage}. We include with it the proof of Proposition \ref{prop:arbitragecalls}. In this section, we only give the proof in the case where $\set=\Omega$. This allows us then to prove Theorem~\ref{thm:bubble_duality} when $\set=\Omega$ which in turn is used to establish the results when $\set\subsetneq \Omega$.
\smallskip\\
\textbf{Step 1}: ``$\exists$ MCM (Market Calibrated Model) $\Longrightarrow$ no WFLVR''\\
First we will show that the existence of a market calibrated model implies no WFLVR. Fix a market calibrated model $\P$ and any $(X_k,\Delta_k)\in \AA_\XX$ and $(X,\Delta)\in \AA_\XX$ such that $\Psi_{X_k,\Delta_k}\to 0$ pointwise on $\set$, $\lim_k\PP(X_k)$ is well defined and $\Psi_{X_k,\Delta_k}\ge \Psi_{X,\Delta}$. Then by Fatou's lemma we know $\liminf_{k}\E_{\P}[\Psi_{X_k,\Delta_k}] \ge \E_{\P}[\liminf_{k}\Psi_{X_k,\Delta_k}]=0$ and hence   
\begin{align*}
\lim_{k}\PP(X_k)=\lim_{k}\E_{\P}[X_k]\ge \liminf_{k}\E_{\P}[\Psi_{X_k,\Delta_k}]\ge \E_{\P}[\liminf_{k}\Psi_{X_k,\Delta_k}]=0.
\end{align*}
\textbf{Step 2}: ``no WFLVR $\Longrightarrow$ Condition \ref{cond:market set-up}"\\
It is straightforward and classical that the absence of a robust uniformly strong arbitrage implies Condition \ref{cond:market set-up} (i)--(ii). Note that since $c_i(\cdot)$ are convex $c_i(0+)$ is well defined. Let $\alpha_i\,:=\,\lim_{K\to \infty}c_i(K)$ which is well-defined by Condition \ref{cond:market set-up} (i) with $\alpha_i\ge 0$ for any $i=1,\ldots,n$. If $\alpha_i>0$ for some $i=1,\ldots,n$ then $(X_k,(0))$ with $X_k=-(\S_i-k)^+$ is a WFLVR  since $X_k\to 0$ pointwise as $k\to \infty$ and $\PP(X_k)=-c_i(k)\to-\alpha_i<0$. If Condition \ref{cond:market set-up} (iv) is violated, then for some $K\in \R_+$ and $i$,   $c_i(0)-c_i(K)< c_{i+1}(0)-c_{i+1}(K)$. Let $X=(\S_i-0)^+-(\S_i-K)^+-(\S_{i+1}-0)^++(\S_{i+1}-K)^+\in \XX_c$, $\Delta_i=\indicator{\S_i<K}$ and $\Delta_j=0$ for $j\neq i$. Then $(X,\Delta)$ is a robust uniformly strong arbitrage since $\Psi_{X,\Delta}\ge 0$ but $\PP(X)< 0$. We conclude that no WFLVR implies Condition \ref{cond:market set-up}. Moreover, absence of a robust uniformly strong arbitrage implies Condition \ref{cond:market set-up} (i), (ii) and (iv).
\smallskip\\
\textbf{Step 3}: ``If $\set=\Omega$ then: Condition \ref{cond:market set-up} $\Longrightarrow$ $\exists$ MCM"\\
Next we show that Condition \ref{cond:market set-up} implies the existence of a market calibrated model when  $\set=\Omega$. It follows from Condition \ref{cond:market set-up} (i), (ii) and (iii) that we can derive probability measures $\vec{\mu}=(\mu_1,\ldots,\mu_n)$ on $(\R_+, \BB(\R_+))$ from the observed prices of call options such that for any $i=1,\ldots,n$ and $K\in \R_+$ 
\begin{align*}
\po_i(K):=\int(x-K)^+\mu_i(\td x)\;\; \text{ and } \;\;\po_i(0)=\int x\mu_i(\td x),
\end{align*}
where $\BB(\R_+)$ is the Borel $\sigma$-algebra of $\R_+$. In fact, due to \citet{breeden1978prices}, $\mu_i$ can be defined via
\begin{align*}
\mu_i([0,K])= 1+c^{\prime}_{i}(K) \quad \text{ for }K\in \R_+. 
\end{align*}
In addition, given $\mu_1,\ldots,\mu_n$ derived from the observed market prices of call options, Strassen's Theorem (\citep{Strassen}) states that Condition \ref{cond:market set-up} (iv) holds if and only if for any convex non-increasing function $\phi:\R_+\to \R$, the sequence $(\int\phi \,d\mu_i)_{i\ge 1}$ is non-decreasing, which is the necessary and sufficient conditions for the existence of a supermartingale on $\R_+^n$ having marginals $\mu_1,\ldots,\mu_n$.  Therefore, when $\set=\Omega$, the absence of WFLVR implies the existence of a market calibrated model which happens if and only if Condition \ref{cond:market set-up} is satisfied.
\smallskip\\
\textbf{Step 4}: ``Additional arguments for Proposition \ref{prop:arbitragecalls}"\\
Given the above steps, to show Proposition \ref{prop:arbitragecalls}, it remains to argue that Condition \ref{cond:market set-up} (i), (ii) and (iv) imply that there is no robust uniformly strong arbitrage when $\set=\Omega$.  Suppose to the contrary that there exists a semi-static strategy $(X,\Delta)$ such that $\Psi_{X,\Delta}\ge 0$ and $\PP(X) < \epsilon < 0$. As $X$ is a finite linear combination of elements of $\XX_c$, we let $K_{\text{max}}$ be the largest among the strikes of call options present in $X$. Then, for any $\delta > 0$ small enough, there exists a sequence of functions $(c_i^{(\delta)})_{i=1}^n$ satisfying Condition \ref{cond:market set-up} (i)--(iv) and such that $|c_i(K)- c^{(\delta)}_i(K)|\le \delta$, for any $i=1,\ldots,n$ and $K \le K_{\text{max}}$. In fact, we can construct $(c_i^{(\delta)})_{i=1}^n$ in the following way. For any $i=1,\ldots,n$, we first define $\tilde{c}_i$ by $\tilde{c}_i(K) = c_i(0) - (1-\frac{\delta i}{2ns_0})(c_i(0) - c_i(K))$ if $c^{\prime}_i(0+)<0$, $\tilde{c}_i(K) = (c_i(0) - \frac{\delta (n+1 -i)}{2nK_{\text{max}}}K)\vee 0$ otherwise. Note that if $c^{\prime}_i(0+)=0$, then $c_j \equiv c_j(0)$ for any $j\ge i$. Then, $|\tilde{c}_i(K) - c_i(K)|\le \delta/2$ for $K\le K_{\text{max}}$ and for $\delta$ sufficiently small $(\tilde{c}_i)_{i=1}^n$ satisfies Condition \ref{cond:market set-up} (i)--(iv) and $\tilde{c}_i(K) - \tilde{c}_i(0)$ is strictly decreasing in $i$ for $K\in (0,K_{\text{max}}]$. Then, for any $i=1,\ldots,n$, we can clearly find convex, decreasing function $c^{(\delta)}_i$ which approximates $\tilde{c}_i$ arbitrarily closely on $[0, K_{\text{max}}]$ and satisfies $c^{(\delta)}_i(0)=\tilde{c}_i(0)$, $\tilde{c}_{i+1}(0)-\tilde{c}_{i+1}\ge c^{(\delta)}_i(0)-c^{(\delta)}_i\ge \tilde{c}_{i}(0)-\tilde{c}_{i}$ and $c^{(\delta)}_i(K)\to 0$ as $K\to \infty$. By the arguments above, with $\PP^{(\delta)}$ corresponding to prices $(c_i^{(\delta)})$, $\PP^{(\delta)}$ and $(c_i^{(\delta)})$ satisfy no WFLVR and hence no robust uniformly strong arbitrage, so $\PP^{(\delta)}(X)\ge 0$. However, we can take $\delta$ small enough so that $|\PP(X)-\PP^{(\delta)}(X)| < \epsilon/2$ which gives the desired contradiction and completes the proof of Proposition \ref{prop:arbitragecalls}.

\subsection{Proof of Theorems \ref{thm:bubble_duality} and \ref{thm:bubble_duality_put}: case where $\set=\Omega$ and $G$ is bounded}\label{subsection:proof of theorem}

We now give the proof of Theorems~\ref{thm:bubble_duality} and \ref{thm:bubble_duality_put} for bounded and upper semi-continuous $G$ in the case where $\set=\Omega$. In this case, we can apply Proposition \ref{proposition:no_arbitrage} or \ref{proposition:no_arbitrage put}. Proof of the general case will be postponed. Since the proof of Theorem \ref{thm:bubble_duality_put} is virtually identical to that of Theorem \ref{thm:bubble_duality}, we only give the proof of Theorem~\ref{thm:bubble_duality}.

\begin{proof}[Proof of Theorems~\ref{thm:bubble_duality}]

We first prove Theorem~\ref{thm:bubble_duality} for bounded and upper semi-continuous $G$ in the case where $\set = \Omega$.  

By Proposition~\ref{proposition:no_arbitrage} in the case where $\set = \Omega$, absence of WFLVR is equivalent to $\MM^-_{\XX_c,\PP,\Omega}\neq 0$, for which to hold Condition \ref{cond:market set-up} is both necessary and sufficient. Following the classical arguments in \citet{breeden1978prices}, by defining probability measures $\mu_i$ on $\R_+$ via
\begin{align}
\mu_i([0,K])= 1+c^{\prime}_{i}(K)  \quad \text{ for }K\in \R_+,
\end{align}
we can encode the market prices $\PP$, or $c_i(K)$, via $(\mu_i)$ with $c_i(K)=\PP((\S_i-K)^+)=\int (s-K)^+\mu_i(\td s)$. Hence $\MM^-_{\XX_c,\PP,\Omega}=\MMNS$.

By Remark~\ref{rk:primallessthandual}, to establish \eqref{eq:thm_bubble_duality_original}, it suffices to show $V_{\XX_c,\PP,\Omega}(G)\le P_{\XX_c,\PP,\Omega}(G)$. 

Define $G_{\Delta}:\R_+^n\to [-\infty,\infty)$ by
$$G_{\Delta}(\S):=G(\S)-\sum_{j=0}^{n-1}\Delta_j(\S_1,\ldots,\S_j)(\S_{j+1}-\S_j).$$
It is clear that if $\Delta_j\in C_{c}(\R^{j}_{+},\R_+)$ for every $j$ then $G_{\Delta}(\S)$ is upper semi-continuous and bounded, and hence satisfies \eqref{eq:Gassumption}. We can deduce that 
\begin{align}
V_{\XX_c,\PP, \Omega}(G)
=& \inf_{(X,\Delta)\in\AA_c\text{ s.t. } \Psi_{X,\Delta}\ge G}\PP(X)
 \label{eq:thm_last step_2}\\
\le & \inf_{\Delta_j\in C_{c}(\R^{j}_{+},\R_+)}\inf\Big\{\PP(X): X\in\AA_c,\text{ s.t. }
X\ge G_{\Delta} \text{ on } \Omega\Big\} \label{eq:thm_last step_3}\\
=& \inf_{\Delta_j\in C_{c}(\R^{j}_{+},\R_+)}\sup_{\pi\in\Pi_{\vec{\mu}}}\Big\{\int_{\R_+^n}G_{\Delta}(s_1,\ldots,s_n)\td \pi(s_1,\ldots,s_n)\Big\}  \label{eq:thm_last step_4}\\
=& \sup_{\pi\in\Pi_{\vec{\mu}}}\inf_{\Delta_j\in C_{c}(\R^n_{+},\R_+)}\Big\{\int_{\R_+^n}G_{,\Delta}(s_1,\ldots,s_n)\td \pi(s_1,\ldots,s_n)\Big\}, \label{eq:thm_last step_5}
\end{align}
where the equality between \eqref{eq:thm_last step_3} and \eqref{eq:thm_last step_4} is guaranteed by Lemma \ref{lemma:Monge-Kantorovich}. To justify the equality between \eqref{eq:thm_last step_4} and \eqref{eq:thm_last step_5} we apply Min-Max Theorem (see Corollary 2 in \citet{min_max_terkelsen1972}) to the compact convex set $\Pi_{\vec{\mu}}$, the convex set $\R_+\times C_{c}(\R_+,\R_+) \times\ldots\times C_{c}(\R^{n-1}_+,\R_+)$, and the function
\begin{align*}
f(\pi,(\Delta_j))=\int_{\R_+^n} G_{\Delta}(s_1,\ldots,s_n)\td \pi(s_1,\ldots,s_n).
\end{align*}
Clearly $f$ is affine in each of the variables. Furthermore, by weak convergence of measures, $f(\cdot,(\Delta_j))$ is upper semi-continuous on $\Pi_{\vec{\mu}}$. Therefore the assumptions of Corollary 2 in \citet{min_max_terkelsen1972} are satisfied.

The last step is to establish the following equality 
\begin{align}\label{eq:thm_last step_6}
\sup_{\pi\in\Pi_{\vec{\mu}}}\inf_{\Delta_j\in C_{c}(\R^{j}_{+},\R_+)}\Big\{\int_{\R_+^n}G_{\Delta}(s_1,\ldots,s_n)\td \pi(s_1,\ldots,s_n)\Big\}=\sup_{\P\in \MMNS}\E_{\P}[G(\S)].
\end{align} 
If $\pi\notin\MMNS$, then by Lemma~\ref{lemma:super_martingale_equivalent}, there is a $\Delta_j\in C_{c}(\R_+^j,\R_+)$ for some $j$ such that
\begin{align*}
B=\int_{\R_+^n}\Delta_j(s_1,\ldots,s_j)(s_{j+1}-s_j)\td \pi(s_1,\ldots,s_n)> 0.
\end{align*}
By scaling, $B$ can be arbitrarily large. Hence, if $\pi\notin\MMNS$, then 
\begin{align}
\inf_{\Delta_j\in C_{c}(\R^{j}_{+},\R_+)}\Big\{\int_{\R_+^n}G_{\Delta}(s_1,\ldots,s_n)\td \pi(s_1,\ldots,s_n)\Big\}=-\infty.\label{eq:thm_penalty_infinity}
\end{align}
Since $G$ is bounded and $\MMNS\neq \emptyset$, $\displaystyle V_{\XX_c,\PP}(G)\ge \sup_{\P\in \MMNS}\E_{\P}[G]>-\infty$.

Therefore, in the LHS of \eqref{eq:thm_last step_6}, it suffices to consider $\pi \in \MMNS$ and then
$$ \inf_{\Delta_j\in C_{c}(\R^{j}_{+},\R_+)}\sum_{j=0}^{n-1}\int\Delta_j(s_1,\ldots,s_j)(s_{j}-s_{j+1})\td \pi = 0.$$
Hence
\begin{align*}
&\sup_{\pi\in\MMNS}\inf_{\Delta_j\in C_{c}(\R^j_{+},\R_+)}\Big\{\int_{\R_+^n}G_{\Delta}(s_1,\ldots,s_n)\td \pi(s_1,\ldots,s_n)\Big\}\\
\le&\sup_{\P\in\MMNS}\E_{\P}[G]+\sup_{\pi\in\MMNS}\inf_{\Delta_j\in C_{c}(\R^{j}_{+},\R_+)}\sum_{j=0}^{n-1}\int\Delta_j(s_1,\ldots,s_j)(s_{j}-s_{j+1})\td \pi=\sup_{\P\in\MMNS}\E_{\P}[G].
\end{align*}
\end{proof}

\subsection{Completing the proof of Proposition \ref{proposition:no_arbitrage}}\label{Appendix:B_continued}

In this section, we complete the proof of Proposition~\ref{proposition:no_arbitrage} in the case where $\set\subsetneq\Omega$.
\smallskip\\
\textbf{Step 5}: ``no WFLVR  $\Longrightarrow$ $\exists$ MCM"\\
It remains to argue that when $\set$ is a closed subset of $\Omega$ such that $\set\subsetneq \Omega$ and Condition \ref{cond:market set-up} is satisfied, no market calibrated model concentrated on $\set$ implies the existence of a WFLVR. In fact, in this case it is a robust uniformly strong arbitrage. Define a lower semi-continuous function $\lambda_{\set}:\R_+^n\to \R$ by 
\begin{align}
\lambda_{\set}(s_1,\ldots,s_n)= \indicator{(s_1,\ldots,s_n)\notin \set}.\label{eq:penalty_prediction_set}
\end{align}
Then we apply Theorem~\ref{thm:bubble_duality} to the prediction set $\Omega$ and $-\lambda_{\set}$ and find that
\begin{align*}
V_{\XX_c,\PP,\Omega}(-\lambda_{\set})=\sup_{\P\in\MM^-_{\XX_c,\PP,\Omega}}\E_{\P}[-\lambda_{\set}]:=\alpha.
\end{align*}
If $\alpha=0$, then there exists a sequence $(\P_k)_{k\in \N}\in \MM^-_{\XX_c,\PP,\Omega}$  such that $\P_k(\set^{\complement})\to 0$. By Lemma \ref{lemma:compactness_super_martingale}, $\MM^-_{\XX_c,\PP,\Omega}$ is compact and closed. Hence $(\P_k)_{k\in \N}$ has a subsequence converging to some $\P\in \MM^-_{\XX_c,\PP,\Omega}$. In fact, by weak convergence of measures, $\P(\set^{\complement})=0$ and hence $\P\in \MM^-_{\XX_c,\PP,\set}$. This shows that, in the absence of a market calibrated model, then $\alpha<0$. Therefore $\alpha<0$ and we can conclude that no market calibrated model concentrated on $\set$ implies the existence of robust uniformly strong arbitrage (and hence WFLVR). Together with the results of Section~\ref{Appendix:B}, this completes the proof of Proposition~\ref{proposition:no_arbitrage}.
\smallskip\\

With the complete proof of Proposition~\ref{proposition:no_arbitrage}, we are now able to get a proof of Theorem~\ref{thm:bubble_duality} in the general case where $\set \neq \Omega$.

\subsection{Completing the proof of Theorems \ref{thm:bubble_duality} and \ref{thm:bubble_duality_put}: case where $\set\subseteq\Omega$}\label{subsection:proof of theorem_continued}

We now complete the proof of Theorems~\ref{thm:bubble_duality} and \ref{thm:bubble_duality_put} for $G$ satisfying \eqref{eq:Gassumption} in the case where $\set \subseteq \Omega$. Again, since they are virtually identical, we only give the proof of Theorem~\ref{thm:bubble_duality} here. 

If \eqref{eq:thm_bubble_duality_original} holds for $G$, then \eqref{eq:thm_bubble_duality_original} is still true for $\tilde{G}=G+X$, for any $X$ taking the form of $a_0+\sum_{i=1}^na_i(S_i-K_i)^+$. Therefore, without loss of generality, we may and will assume that $G$ is bounded from above.

Recall from \eqref{eq:penalty_prediction_set} that $\lambda_{\set}(s_1,\ldots,s_n)= \indicator{(s_1,\ldots,s_n)\notin \set}$ is bounded and lower semi-continuous and hence $G\vee (-N)-N\lambda_{\set}$ is bounded and upper semi-continuous for each $N\in \N$. We also notice that $\dualpc(G)\le V_{\XX_c,\PP, \Omega}(G\vee (-N)-N\lambda_{\set})$ for each $N\in \N$, as a super-replicating portfolio of $G\vee (-N)-N\lambda_{\set}$ on $\Omega$ naturally super-replicates $G$ on $\set$. Thus 
\begin{align*}
\dualpc(G)\le& \inf_{N\ge 0}V_{\XX_c,\PP, \Omega}(G\vee (-N)-N\lambda_{\set})\\
 =&\inf_{N\ge 0}P_{\XX_c,\PP, \Omega}(G\vee (-N)-N\lambda_{\set}) =  \inf_{N\ge 0}\sup_{\P\in \MMNS}\E_{\P}[G\vee (-N)-N\lambda_{\set}].
\end{align*}
Define $f_N:\MMNS \to (-\infty, \infty)$ by 
$f_N(\P) = \E_{\P}[ G\vee (-N) - N\lambda_{\set}]$. Note that $f_N$ is upper semi-continuous on $\MMNS$ by weak convergence of measures and $f_N\ge f_{N+1}$ for every $N\in \N$. Hence, applying Min-Max Theorem (see Corollary 1 in \citet{min_max_terkelsen1972}) to the compact convex set $\MMNS$ and $(f_{N})_{N\in \N}$, we have 
\begin{align*}
\inf_{N\ge 0}\sup_{\P\in \MMNS}\E_{\P}[G\vee (-N)-N\lambda_{\set}]
=\sup_{\P\in \MMNS}\inf_{N\ge 0}\E_{\P}[G\vee (-N)-N\lambda_{\set}].
\end{align*}
Define $G_{\set}:\Omega\to [-\infty, \infty)$ by $G_{\set} = G$ on $ \set$ and $-\infty$ elsewhere. Note that $G_{\set}$ is the point-wise limit of $G\vee (-N) - N\lambda_{\set}$ as $N\to \infty$. Then by Fatou's lemma, 
\begin{equation*}
 \sup_{\P\in \MMNS}\inf_{N\ge 0}\E_{\P}[G\vee (-N)-N\lambda_{\set}]\le \sup_{\P\in \MMNS}\E_{\P}[G_{\set}] = \sup_{\P\in \MMUS}\E_{\P}[G]. 
\end{equation*}
Therefore we have $\dualpc(G) \le \primalpc(G)$, which together with Remark~\ref{rk:primallessthandual} leads us to conclude that
$$ \dualpc(G) = \primalpc(G). $$

\subsection{Proof of Theorem \ref{thm:no duality put}}\label{sec:proof_nmarg_duality}

\begin{proof}
Given a semi-static super-replicating strategy $(X,\Delta)$ we have, by definition, 
\begin{align}
X(s_1,\ldots,s_n)+\sum_{i=0}^{n-1}\Delta_i(s_1,\ldots,s_i)(s_{i+1}-s_i)\ge G(s_1,\ldots,s_n),\quad (s_1,\ldots,s_n)\in\set.\label{eq:case study 3 super-replicating}
\end{align}
We start with the following\\
\textbf{Claim}: \emph{If $(X,\Delta)$ is a semi-static super-replicating strategy of $G$ on $\set$, then $\Delta_j\ge \beta_j$ for any $i=0,\ldots,n-1$.}  \\
We prove the claim by induction.
When $j=n-1$, we fix $\vec{s}_{n-1}:=(s_1,\ldots,s_{n-1})$. Letting $s_n\in \set(\vec{s}_{n-1},n):=\{x\,:\,(s_1,\ldots,s_{n-1},x)\in\set\}$ go to infinity, it follows from \eqref{eq:case study 3 super-replicating} that
\begin{align*}
\Delta_{n-1}(s_1,\ldots,s_{n-1})\ge \limsup_{x\to\infty,\,x\in\set(\vec{s}_{n-1},n)}\frac{G(s_1,\ldots,s_{n-1},x)}{x}.
\end{align*}
This, together with $\Delta_{n-1}\ge 0$, yields 
\begin{align*}
\Delta_{n-1}(s_1,\ldots,s_{n-1})\ge&\, \limsup_{x\to\infty}\Big\{\frac{G(s_1,\ldots,s_{n-1},x)}{x}\indicators{\set}(s_1,\ldots,s_{n-1},x)\Big\}\vee 0 \\
=&\,\beta_{n-1}(s_1,\ldots,s_{n-1}).
\end{align*}
Now suppose the claim holds for $j=i+1$ with $i\le n-2$. Fix $\vec{s}_{n-1}:=(s_1,\ldots,s_{i},s_{i+2},\ldots,s_{n})$ and denote $\set(\vec{s}_{n-1},i+1):=\{x\,:\,(s_1,\ldots,s_i,x,s_{i+2},\ldots,s_{n})\in\set\}$. If $\set(\vec{s}_{n-1},i+1)$ is unbounded, then by taking $x\in \set(\vec{s}_{n-1},i+1)$ to infinity, \eqref{eq:case study 3 super-replicating} implies 
\begin{align*}
\Delta_{i}(s_1,\ldots,s_{i})\ge \limsup_{x\to\infty,\,x\in \set(\vec{s}_{n-1},i+1)}\Big(\Delta_{i+1}(s_1,\ldots,s_i,x)+\frac{G(s_1,\ldots,s_i,x,s_{i+2},\ldots,s_n)}{x}\Big).
\end{align*}
If the right hand side is non-negative, then
\begin{align*}
&\limsup_{x\to\infty,\,x\in \set(\vec{s}_{n-1},i+1)}\Big(\Delta_{i+1}(s_1,\ldots,s_i,x)+\frac{G(s_1,\ldots,s_i,x,s_{i+2},\ldots,s_n)}{x}\Big)\\
=&\limsup_{x \to \infty}\bigg(\Big(\Delta_{i+1}(s_1,\ldots,s_i,x)+\frac{G(s_1,\ldots,s_i,x,s_{i+2},\ldots,s_n)}{x}\Big)\indicators{\set}(s_1,\ldots,s_i,x,s_{i+2},\ldots,s_n)\bigg).
\end{align*}
Hence, as $\Delta_{i}\ge 0$, when $\set(\vec{s}_{n-1},i+1)$ is unbounded, we have 
\begin{align}
\Delta_{i}(s_1,\ldots,s_{i})\ge& \limsup_{x \to \infty}\bigg(\Big(\Delta_{i+1}(s_1,\ldots,s_i,x)\nonumber\\
&+\frac{G(s_1,\ldots,s_i,x,s_{i+2},\ldots,s_n)}{x}\Big)\indicators{\set}(s_1,\ldots,s_i,x,s_{i+2},\ldots,s_n)\bigg)\vee 0\nonumber\\
\ge& \limsup_{x \to \infty}\bigg(\Big(\beta_{i+1}(s_1,\ldots,s_i,x)\nonumber\\
&+\frac{G(s_1,\ldots,s_i,x,s_{i+2},\ldots,s_n)}{x}\Big)\indicators{\set}(s_1,\ldots,s_i,x,s_{i+2},\ldots,s_n)\bigg)\vee 0.\label{eq:delta_beta_1}
\end{align}
On the other hand, when $\set(\vec{s}_{n-1},i+1)$ is bounded, we notice that 
\begin{align*}
\limsup_{x \to \infty}\bigg(\Big(\beta_{i+1}(s_1,\ldots,s_i,x)+\frac{G(s_1,\ldots,s_i,x,s_{i+2},\ldots,s_n)}{x}\Big)\indicators{\set}(s_1,\ldots,s_i,x,s_{i+2},\ldots,s_n)\bigg)=0.
\end{align*}
Therefore, the inequality in \eqref{eq:delta_beta_1} is true in either case. In addition, as it holds for any $s_1,\ldots,s_i,s_{i+2},\ldots,s_n\in \R_+$, we can conclude that 
\begin{align*}
\Delta_{i}(s_1,\ldots,s_{i})\ge& \sup_{s_{i+2},\ldots,s_{n}\in\R_+}\limsup_{x \to \infty}\bigg(\Big(\Delta_{i+1}(s_1,\ldots,s_i,x)\\
&+\frac{G(s_1,\ldots,s_i,x,s_{i+2},\ldots,s_n)}{x}\Big)\indicators{\set}(s_1,\ldots,s_i,x,s_{i+2},\ldots,s_n)\bigg)\vee 0\\
=&\beta_{i}(s_1,\ldots,s_{i}), \qquad \text{ for any }s_1,\ldots,s_i\in\R_+. 
\end{align*}
This ends the induction and the proof of the claim.

It follows from the claim above that for any $(X,\Delta)\in \AA_p$ that super-replicates $G$ on $\set$ and any $\P\in \MMUS$
\begin{align*}
\E_{\P}[G]\le&\E_{\P}\Big[X(\S)+\Delta_0(\S_{1}-s_0)+\sum_{i=1}^{n-1}\Delta_i(\S_1,\ldots,\S_{i})(\S_{i+1}-\S_i)\Big]\\
\le&\E_{\P}\Big[X(\S)+\beta_0(\S_{1}-s_0)+\sum_{i=1}^{n-1}\beta_i(\S_1,\ldots,\S_{i})(\S_{i+1}-\S_i)\Big],
\end{align*}
which implies that 
\begin{align}
V^{(p)}_{\vec{\mu},\set}(G)\ge\sup_{\P\in\MMUS}\E_{\P}\Big[G-\beta_0(\S_{1}-s_0)-\sum_{i=1}^{n-1}\beta_i(\S_1,\ldots,\S_{i})(\S_{i+1}-\S_i)\Big].\label{eq:multiple_put_correction}
\end{align}

On the other hand, we denote $\ZZ$ by the set of $(\Delta_j)_{j=0}^{n-1}\in \R_+\times \BB_b(\R_+,\R_+) \times\ldots\times \BB_b(\R^{n-1}_+,\R_+)$ such that 
$$ G_{\Delta}(\S):=G(\S)-\sum_{i=0}^{n-1}\Delta_i(\S_1,\ldots,\S_i)(\S_{i+1}-\S_i)$$ is upper semi-continuous and bounded from above on $\set$, where $\BB_b(\R_+^d, \R_+)$ is the set of bounded measurable functions $f:\R_+^d\to\R_+$. Note that $\ZZ$ is a convex subset of $\R_+\times \BB_b(\R_+,\R_+) \times\ldots\times \BB_b(\R^{n-1}_+,\R_+)$. Then we can apply Min-Max Theorem (Corollary 2 in \citet{min_max_terkelsen1972}) to the compact convex set $\MMUS$, $\ZZ$ and the function
\begin{align*}
f(\pi,(\Delta_j))=\int \left(G(s_1,\ldots,s_n)-\Delta_0(s_{1}-s_0)-\sum_{i=1}^{n-1}\Delta_i(s_1,\ldots,s_i)(s_{i+1}-s_i)\right)\,\td\pi(s_1,\ldots,s_n).
\end{align*}
Clearly $f$ is affine in each of the variables. Furthermore, since for any fixed $(\Delta_j)_{j=0}^{n-1} \in \ZZ$ the term $\int_{\R^n_+\setminus [0,a]^n}G_{\Delta}\td \pi$ converges to $0$ uniformly in $\pi\in \MMUS$ as $a\to \infty$, it follows from the portemanteau theorem that $f(\cdot\,,(\Delta_j))$ is upper semi-continuous on $\MMUS$. Therefore the assumptions of Corollary 2 in \citet{min_max_terkelsen1972} are satisfied. Then we find that
\begin{align}
V^{(p)}_{\vec{\mu},\set}(G)
=&\inf\Big\{\PP(X): (X,\Delta)\in \AA_p \text{ s.t. } \Psi_{X,\Delta}\ge G \text{ on } \set\Big\} \label{eq:case study 3 put_duality_2}\\
\le&\inf_{\Delta\in \ZZ}
\inf\Big\{\PP(X)\,:\, (X,\tilde{\Delta})\in \AA_p \text{ s.t. }\Psi_{X,\Delta+\tilde{\Delta}}\ge G\text{ on } \set\Big\}\label{eq:case study 3 put_duality_3}\\
=&\inf_{\Delta\in \ZZ}\inf\Big\{\PP(X)\,:\, (X,\tilde{\Delta})\in \AA_p \text{ s.t. } \Psi_{X,\tilde{\Delta}}(\S)\ge G_{\Delta}(\S) \text{ on } \set\Big\}\label{eq:case study 3 put_duality_4}\\
=&\inf_{\Delta\in \ZZ}\sup_{\P\in \MMUS}\E_{\P}\big[G_{\Delta}(\S)\big] =\sup_{\P\in \MMUS}\inf_{\Delta\in \ZZ}\E_{\P}\big[G_{\Delta}(\S)\big]\nonumber
\end{align}
where the inequality between \eqref{eq:case study 3 put_duality_2} and \eqref{eq:case study 3 put_duality_3} is by restricting the delta hedging terms to a smaller set and the equality between \eqref{eq:case study 3 put_duality_3} and \eqref{eq:case study 3 put_duality_4} follows from Theorem \ref{thm:bubble_duality_put}. 

To conclude, from the assumption we know there exists $(\beta^{(N)})_{N\ge 1}\in\ZZ$ such that $G_{\beta^{(N)}}(\S)=G(\S)-\beta^{(N)}_0(\S_{1}-s_0)-\sum_{i=1}^{n-1}\beta^{(N)}_i(\S_1,\ldots,\S_{i})(\S_{i+1}-\S_i)$ is upper semi-continuous, bounded from above on $\set$ and $G_{\beta^{(N)}}\to G_{\beta}$ pointwise as $N\to \infty$. Hence, by Fatou's Lemma  
\begin{align*}
\limsup_{N\to \infty}\E_{\P}[G_{\beta^{(N)}}(\S)]\le \E_{\P}\Big[G-\sum_{i=1}^{n-1}\beta_i(\S_1,\ldots,\S_i)(\S_{i+1}-\S_i)-\beta_0(\S_1-s_0)\Big]
\end{align*}
holds for any $\P\in \MMUS$ and therefore we have
$$\sup_{\P\in \MMUS}\inf_{\Delta\in \ZZ}\E_{\P}\big[G_{\Delta}(\S)\big]\le\sup_{\P\in \MMUS}\E_{\P}\Big[G-\sum_{i=1}^{n-1}\beta_i(\S_1,\ldots,\S_i)(\S_{i+1}-\S_i)-\beta_0(\S_1-s_0)\Big],$$
which leads us to conclude that
\begin{align*}
V^{(p)}_{\vec{\mu},\set}(G)\le\sup_{\P\in\MMUS}\E_{\P}\Big[G-\beta_0(\S_{1}-s_0)-\sum_{i=1}^{n-1}\beta_i(\S_1,\ldots,\S_{i})(\S_{i+1}-\S_i)\Big].
\end{align*}

\end{proof}

\bibliographystyle{apalike}

\bibliography{bib}

\end{document}